\documentclass{article}

\usepackage[preprint]{neurips_2021}

\usepackage[utf8]{inputenc} 
\usepackage[T1]{fontenc}    
\usepackage{url}            
\usepackage{booktabs}       
\usepackage{amsfonts}       
\usepackage{nicefrac}       
\usepackage{microtype}      
\usepackage{xcolor}         

\usepackage[ruled]{algorithm2e} 

\usepackage[utf8]{inputenc}
\usepackage{amsmath}
\usepackage{amsthm}
\usepackage{amssymb}
\usepackage{amsfonts}
\usepackage{bbm}
\usepackage{color}
\usepackage{mathtools}

\usepackage{enumitem}

\usepackage{hyperref}       

\usepackage[mode=buildnew]{standalone}

\usepackage{caption}
\usepackage{subcaption}

\usepackage{apxproof} 

\newtheorem{theorem}{Theorem}
\newtheorem{lemma}[theorem]{Lemma}
\newtheorem{proposition}[theorem]{Proposition} 

\newtheorem{definition}[theorem]{Definition}

\newtheorem{fact}[theorem]{Fact}

\newcommand{\brb}[1]{\bigl(#1\bigr)}

\newif\ifPrint 

\newtheorem{innercustomthm}{Theorem}
\newenvironment{customthm}[1]
  {\renewcommand\theinnercustomthm{#1}\innercustomthm}
  {\endinnercustomthm}
  
\newtheorem{innercustomprop}{Proposition}
\newenvironment{customprop}[1]
  {\renewcommand\theinnercustomprop{#1}\innercustomprop}
  {\endinnercustomprop}
  
\newtheorem{innercustomlemma}{Lemma}
\newenvironment{customlemma}[1]
  {\renewcommand\theinnercustomlemma{#1}\innercustomlemma}
  {\endinnercustomlemma}

\SetKwRepeat{Do}{do}{while}

\definecolor{burntorange}{rgb}{0.85, 0.35, 0.1}
\definecolor{charcoal}{rgb}{0.21, 0.27, 0.31}
\definecolor{coolblack}{rgb}{0.0, 0.28, 0.49}
\definecolor{burntgreen}{rgb}{0.05, 0.45, 0.27}
\definecolor{burntblue}{rgb}{0.05, 0.27, 0.8}

\newcommand{\offerNext}{\textsc{offerNext}}
\newcommand{\bestAmong}{\textsc{bestAmong}}
\newcommand{\weakestMatch}{\textsc{weakestMatch}}
\newcommand{\GSalgShort}{\textsc{GS}} 

\newcommand{\countMatches}{\textsc{countMatches}}

\title{Finding Stable Matchings in PhD Markets with Consistent Preferences and Cooperative Partners}

\author{
   \textbf{Maximilian Mordig}\\
MPI for Intelligent Systems and ETH Zürich
\And
\textbf{Riccardo Della Vecchia}\\
  Artificial Intelligence Lab, Institute for Data Science \& Analytics\\
    Bocconi University, Milano, Italy
  \And
 \textbf{Nicol\`o Cesa-Bianchi}\\
 Dipartimento di Informatica \& DSRC, \\
 Universit\`a degli Studi di Milano, Milano, Italy
 \And 
 \textbf{Bernhard Sch\"olkopf}\\
MPI for Intelligent Systems and ETH Zürich
 }

\begin{document}
\maketitle

\begin{abstract}
We introduce a new algorithm for finding stable matchings in multi-sided matching markets.
Our setting is motivated by a PhD market of students, advisors, and co-advisors, and can be generalized to supply chain networks viewed as $n$-sided markets. 
In the three-sided PhD market, students primarily care about advisors and then about co-advisors (consistent preferences), while advisors and co-advisors have preferences over students only (hence they are cooperative). A student must be matched to one advisor and one co-advisor, or not at all. 
In contrast to previous work, advisor-student and student-co-advisor pairs may not be mutually acceptable (e.g., a student may not want to work with an advisor or co-advisor and vice versa).
We show that three-sided stable matchings always exist, and present an algorithm that, in time quadratic in the market size (up to log factors), finds a three-sided stable matching using any two-sided stable matching algorithm as matching engine.
We illustrate the challenges that arise when not all advisor-co-advisor pairs are compatible.
We then generalize our algorithm to $n$-sided markets with quotas and show how they can model supply chain networks.
Finally, we show how our algorithm outperforms the baseline given by \citep{danilov2003existence} in terms of both producing a stable matching and a larger number of matches on a synthetic dataset.
\end{abstract}


\section{Introduction}
Matching problems, a fundamental topic in economics and game theory, naturally arise in both online labor markets and markets rooted in the physical world. Recently, a number of works have started to use machine learning tools to analyze two-sided matching markets in which agents have imperfect knowledge about their own preferences \citep{liu2020competing,cen2021regret}.
In these studies, the notion of stable matching plays a central role, as it provides a touchstone for measuring the performance of a learning agent. Hence, understanding stability, and devising efficient algorithms to find stable matchings, is a crucial step in the development of learning applications to matching problems.

Two-sided matching markets, however, are not sufficient to model many interesting applications. Motivated by a real-world three-sided PhD market of students, advisors, and co-advisors, we introduce the PhD algorithm, a new matching algorithm that runs in time quadratic in the market size (up to log factors) and finds a family of three-sided stable matchings using any two-sided stable matching algorithm as matching engine. Our algorithm works in settings where students have consistent preferences, and advisors and co-advisors only care about the student. Unlike previous setups \citep{danilov2003existence, huang2007two}, we require complete matches rather than partial ones, and students may not find all advisors and co-advisors acceptable and vice versa. 
We also show that our algorithm can find stable matchings in the presence of quotas and in more general 
$n$-sided markets that can model supply chain networks. 
In our analysis, we discuss the challenges related to incompatible advisor pairs, and present a counterexample for the algorithm presented in \citep{zhong2019Cooper}. Finally, we evaluate our algorithm on a synthetic dataset and improve upon the baseline by \cite{danilov2003existence, huang2007two}.

Two-sided stable matching is a classical combinatorial problem consisting of $N$ men and $N$ women where each person has ordinal preferences over all the persons of the opposite sex and they have to be matched together.%
\footnote{For clarity, we keep the terminology used in the literature. Please consider \textit{men} and \textit{women} as place holders for any two parties that need to be matched.}
The goal is to find a stable matching such that any man and woman that are not matched together do not both prefer each other to their actual partners.
In their landmark paper, Gale and Shapley showed that a stable matching always exists, and give an algorithm, the Gale--Shapley (GS) algorithm, to construct it \citep{gale1962college}.
The problem can be extended to the case where the number of men and women are not identical, and a person may prefer to stay single rather than match with an unacceptable partner (in which case this person is said to have incomplete preferences), see \citep{itoga1978upper, knuth1997stable, wilson1972analysis}. A (valid) matching must match each person to a person from the opposite sex \emph{or} himself.
The GS matching is man-optimal and woman-pessimal in the lattice of stable matchings \citep{gale1962college, knuth1997stable}. 
The problem was extended to many-to-one (known as hospital-residents or college admissions problem) and many-to-many matchings \citep{gale1962college, roth1991natural, roth1992two}. Other works explored more general concepts like group stability, and the connection to the game-theoretic core \citep{blair1988lattice, echenique2004theory}.

A natural extension are three-sided markets such as the family/man--woman--child market, where each man must be matched to exactly one woman and one child (or stay single), and likewise for women and children \citep{knuth1997stable}.
When each side has preferences over couples from the other two sides, a stable matching does not always exist \citep{alkan1988nonexistence}.
In fact, deciding whether a stable matching exists is NP-complete \citep{ng1991three}. 
Preferences over couples are said to be {\em consistent} if they are a product order (lexicographic order). That is, a man primarily cares about the woman and secondarily about the child, analogously for women and children.
Again, deciding whether stable matchings exist is NP-complete when preferences are consistent and complete \citep{huang2007two}.
In the special case when men primarily care about the woman and then about the child, women primarily care about the man and then about the child, and everyone finds everyone else acceptable (completeness), then \citep{danilov2003existence, huang2007two} show that stable three-sided matching exist.
First, compute a stable matching between men and women (e.g., via the GS algorithm), and then match each couple of man and woman to a child, where the couple has the man's preferences over children.
The assumption of complete preferences is very strong, e.g., a woman may not want to match to every other man or cannot have preferences over all men. It might then happen that a man is matched to a child, but to no woman. We explicitly want to avoid such partial matches, e.g., a man with a child should not be matched to the child if he cannot find any woman. The above procedure cannot be directly adapted.
Since the motivation of this project came from a matching procedure for a large pan-continental PhD program, we adopt the terminology of a PhD market with advisors, students and co-advisors. In this setting, we use consistent and cooperative preferences.

A closely related problem to the one presented here is the hospital-residents problem with couples (HPC), where each couple of medical students is a single unit, and must be matched to exactly two hospitals or stay unmatched \citep{roth1984evolution, manlove2017almost}.
Deciding whether stable matchings exist in this case is again NP-complete \citep{ronn1990np}, therefore research has focused on finding good heuristics to find (almost) stable matchings \citep{biro2011stable, biro2013matching}.
Our variant of the three-sided problem $(A, S, C)$ can be transformed into a two-sided problem $(A \cup C, S)$, with the property that pairs outside of $A \times C$ are not acceptable to any $s \in S$, and each student must be matched to exactly two or zero elements of $A \cup C$.
Hence, it is a special case of the two-sided HPC matching problem, where each student can match to zero or two elements of $(A \cup C) \times (A \cup C)$.
Whereas the general HPC problem is NP-complete, and stable matchings need not exist even when students have consistent preferences \citep{mcdermid2010keeping}, we show that our variant always admits stable matchings.
Finally, note that our problem is not a many-to-many matching problem with substitutable preferences as in the UK medical market \citep{roth1991natural} for which \citep{roth1991natural, echenique2004theory} provide algorithms: students who are not matched to an advisor cannot be matched to a co-advisor, so their preferences cannot be substitutable (unless these preferences are empty).

\section{Two-sided matching markets}\label{sec:twosidedMarketsShort}
In this introductory section, we summarize the relevant literature on two-sided markets, see \citep{gale1962college, roth1984evolution, sotomayor1990two} for a more in-depth treatment.
We present the setting with strict preferences and without quotas. 

A \emph{two-sided many-to-many matching market}, denoted by $(M, W, P)$, consists of a set of ``men'' $M$, a set of ``women'' $W$, and preferences $P$, where $P(m)$ is a total order on $W \cup \{ m \}$ for each man $m \in M$ denoted by $<_m, \leq_m$, and $P(w)$ is a total order on $M \cup \{ w \}$ for each woman $w \in W$ denoted by $<_w, \leq_w$.
Person $p$ stays single if he is matched to himself rather than someone on the other side.
If $p_1 <_p p$, it means that person $p$ prefers to be self-matched rather than to be matched to $p_1$.
If $p_1 >_p p$, person $p_1$ is acceptable to $p$.
A person has complete preferences if he prefers everyone else to himself.
A matching on this market is a set $\mu \subset M \times W$ which matches a person at most once. Let $\mu(p)$ denote the match partner of $p$. If $p$ is not matched, we define $\mu(p) = p$ self-matched.
A matching is \emph{individually rational} if each person prefers his match partner to being matched to himself, i.e., $\mu(p) \geq_p p \; \forall p \in M \cup W$.
A matching is blocked by the pair $(m, w) \in M \times W$ if $w >_m \mu(m)$ and $m >_w \mu(w)$.
In words, if $m$ and $w$ prefer each other to their current match partners respectively, they have a strong incentive to match instead and disrespect the matching.
A matching is \emph{stable} if it is individually rational and not blocked by any pair. 
Stability is of utmost importance in practical applications since the persons will otherwise find their own arrangements \citep{roth1984evolution}.
A matching $\mu$ is \emph{man-optimal} if $\mu(m) \geq_m \lambda(m) \; \forall m \in M$ for any stable matching $\lambda$. Woman-optimality is analogous and pessimality inverts the inequalities.

It is not clear a priori whether a stable matching always exists. The Gale--Shapley (GS) algorithm constructively shows that a stable matching always exists.
The GS algorithm works by letting men propose to women and women conditionally accept unless they get an offer from a better man later on. It starts with all men unmatched and all women matched to themselves.
As long as a man is unmatched, consider any unmatched man. He proposes to his next most preferred woman he has not proposed to already.
If a man has proposed to all of his acceptable women, he proposes to himself instead and accepts (so he is self-matched/single).
If the woman prefers the man to her current partner, she disengages from her old partner (a man or herself), leaving her old partner unmatched again. She engages/matches with this new man.
The algorithm stops once each man is matched (to a woman or himself).
The matched men and women are ``married''. The matching is the same independently of how a man is picked among the free men.
The GS algorithm produces a stable matching in time $\mathcal{O}(|M| |W|)$.
Moreover, the matching is man-optimal and woman-pessimal among all stable matchings, also known as man-man/woman-woman coincidence of interest and man-woman conflict of interest \citep{echenique2004theory, roth1985conflict}.
By swapping the roles of men and women, it returns a man-pessimal and woman-optimal matching.
The theory and GS algorithm can be extended to non-strict preferences, when persons are indifferent between persons of the other side (Appendix~\ref{app:missingProofs}). Moreover, each man and woman can individually specify a quota that limits the maximum number of match partners they can have, e.g., in polygamous relationships (Appendix~\ref{app:quotas}). It is possible to ensure that the same man and woman do not get matched more than once by having the man only propose to women he did not propose to already whenever he has free spots.

We will use the following result to show an interesting invariance property of the matching algorithm that we propose in three-sided markets and that carries over to $n$-sided markets.
Using optimality, one can show that the set of matched persons is the same among all stable matchings, see Appendix~\ref{app:missingProofs}.
\begin{proposition}\label{prop:twoSidedSameMatched}
Let ($M$, $W$, $P$) a two-sided market with strict preferences. The set of matched persons is the same in any stable matching.
\end{proposition}

\section{Three-sided matching markets}\label{sec:threeSidedMarketsShort}
We first introduce three-sided markets with special preferences and then show how to find stable matchings.
We adopt the terminology of PhD markets where each student must be matched to exactly one advisor in $A$ and one co-advisor in $C$, or stay single.

A \emph{three-sided matching market}, denoted by $(A, S, C, P)$, consists of advisors $A$, students $S$ and co-advisors $C$ with preferences $P$. 
Since stable matchings may not exist for arbitrary preferences \citep{ng1991three, lam2019existence, huang2007two, ronn1990np, mcdermid2010keeping}, we consider the special case where advisors/co-advisors only care about students and they are indifferent between co-advisors/advisors, i.e., they are \emph{cooperative} partners. 
This is reasonable if the student works with the advisor and co-advisor on separate projects, or if an advisor is happy to work with any co-advisor that the student would work with.
The same professor may independently advise and co-advise students, so the number of students he advises has no influence on the number of students he can co-advise.
In addition, students have \emph{consistent preferences} in the sense that they primarily care about the advisor and less about the co-advisor. 
More precisely, student $s$ has separate preferences over advisors and co-advisors, $\leq_s^A$ and $\leq_s^C$, and his compound preferences between advisor pairs $(a, c), (\tilde{a}, \tilde{c}) \in (A \times C) \cup \{ (s, s) \}$ are $(a, c) >_s (\tilde{a}, \tilde{c}) \iff a >_s^A \tilde{a} \lor (a = \tilde{a} \land c >_s^C \tilde{c})$, where $(s, s)$ means that student $s$ stays single.
We drop the superscripts $A$ and $C$ since it will be clear from the context.
In other words, we have preferences over the two two-sided markets $(A, S)$ and $(S, C)$: $A \leftrightarrow S \leftrightarrow C$. Importantly, the two-sided preferences $A \leftrightarrow S \leftrightarrow C$ need not be complete, e.g., a student may not want to work with all advisors $A$. For this reason, we design an extension of the setting in \citep{danilov2003existence}. 

A matching is a set $\mu \subset A \times S \times C$ which matches the same person at most once. Let $\mu(p)$ denote the couple formed by the match partners of $p$. If $p$ is not matched, set $\mu(p) = (p, p)$.
Given a matching $\mu$ on a two-sided market $(M, W, P)$ and subsets $M_1 \subset M, W_1 \subset W$,  we denote with $P|_{(M_1, W_1)}$ the restricted preferences where persons not in $M_1 \cup W_1$ are removed from the preference lists. The matching $\mu|_{(M_1, W_1)}$ is the matching restricted to the market $(M_1, W_1, P|_{(M_1, W_1)})$, where persons who are matched to persons not in $M_1 \cup W_1$ are matched to themselves instead. We define analogous restrictions for three-sided markets.
The three-sided matching $\mu$ is individually rational if the restricted two-sided matchings $\mu|_{(A, S)}$ and $\mu|_{(S, C)}$ are individually rational.
From now on, we consider only individually rational matchings. 
A triple $(a, s, c) \in A \times S \times C$ is mutually acceptable if $s \geq_a a, a \geq_s s, c \geq_s s, s \geq_c c$.
A matching is stable if it is not blocked by any mutually acceptable triple.
The matching $\mu$ is blocked by the mutually acceptable $(a, s, c)$ if i) student $s$ is unmatched in $\mu$ and $(a, s)$ blocks $\mu|_{(A, S)}$ and $(s, c)$ blocks $\mu|_{(S, C)}$, or ii) student $s$ is matched in $\mu$ and $(a, s)$ blocks $\mu|_{(A, S)}$ or $(s, c)$ blocks $\mu|_{(S, C)}$.
In other words, either $s$ is unmatched and all of $(a, s, c)$ can improve their matches on the two-sided markets, or $s$ is matched and either $(a, s)$ can improve or $(s, c)$ can improve.
Again, stability is a crucial property because persons will not stick to matches they are assigned in unstable matchings.
This coincides with the definitions in \citep{zhang2019Hybrid, huang2007two, manlove2017almost}, see Appendix~\ref{app:missingProofs}. In fact, these definitions show that advisors and co-advisors can equivalently be treated symmetrically.
When $(a, s, c)$ is a blocking triple, $(a, s)$ either blocks $\mu|_{(A, S)}$, or $s$ is matched and $(s, c)$ blocks $\mu|_{(S_m, C)}$. 

\subsection{PhD algorithm}\label{sec:phdAlgorithm}
In this section we prove that stable matchings exist for three-sided matching markets. 
The idea is to iteratively apply the GS algorithm to the two two-sided markets and to greedily remove students when they cannot find an advisor and co-advisor simultaneously.
Initially, all students $S$ participate. Students and advisors are matched in a GS algorithm. The matched students $S_m \subset S$ participate in a GS algorithm together with co-advisors. Students $S_u$ who participate in the second market with co-advisors and don't find a co-advisor (partial advisor-student match) are removed. The remaining students $S \setminus S_u$ can participate in the next round, including students who did not find an advisor. This is repeated until no students are removed. The removed students as well as the students who could not find an advisor at the last iteration are matched to themselves/stay single.
Intuitively, students who are removed are ``bad'' and would also not find a co-advisor in later iterations.
We call this algorithm the PhD algorithm and it is depicted in Algorithm~\ref{alg:phdAlgorithm}.
\begin{algorithm}[ht]
\DontPrintSemicolon
\SetAlgoLined
\KwIn{PhD market $(A, S, C, P)$}
\KwOut{stable matching $\mu$}
strictify the preferences $P_{AS}$ on $(A, S)$ and $P_{SC}$ on $(S, C)$ \;
\Do{$S_u\neq\emptyset$ }{
$\mu_{AS} \gets \GSalgShort(M = S, W = A, P_{AS})$\;
let $S_m$ the matched students in $\mu_{AS}$\;
$\mu_{SC} \gets \GSalgShort(M = S_m, W = C, P_{SC})$\;
let $S_u \subset S_m$ the unmatched students in $\mu_{SC}$\;
$S \gets S \setminus S_u$
}
students in $S$ are matched to respective advisor in $\mu_{AS}$ and co-advisor in $\mu_{SC}$, the others remain single\;
\Return{$\mu$}
\caption{PhD Market - PhD algorithm}
\label{alg:phdAlgorithm}
\end{algorithm}
The function $\GSalgShort(M, W, P)$ stands for the GS algorithm where men $M$ propose to women $W$ under preferences $P$. 
The algorithm terminates because the GS algorithm itself terminates and the number of students $|S|$ decreases at each iteration, unless $S_u$ is empty and the while-loop ends.
The algorithm returns a valid matching since no student is partially matched. Indeed, student $s$ is matched to an advisor if and only if $s \in S_m$ at the final iteration. Since $S_u$ is empty at the final iteration, $s \in S_m$ is also matched to a co-advisor.
Preferences are initially strictified by breaking (all) ties and the matching is also stable under the original non-strictified preferences (Appendix~\ref{app:missingProofs}). If preferences are not strictified initially, the matching may not be stable or the number of matches between stable matchings may differ.
So we assume that preferences are strict from now on.

To prove stability, we establish a few properties first. 
One iteration refers to one iteration of the PhD algorithm after which unmatched students are removed. 
Superscripts refer to the iteration number: $\mu_{AS}^{(\tau)}$ is the matching at iteration $\tau = 1, \dots, T$ and $S^{(\tau)}$ are the remaining students at the start of iteration $\tau$. $T$ is the total number of iterations.
We first assume that the proposing side is fixed over iterations in each of the two-sided markets.
To understand how matches evolve over iterations on the two-sided markets, we start by studying the matchings produced by the GS algorithm when people on one side are removed.
When men propose and new men are added to the market, the existing men cannot get better matches and the women cannot get worse matches since there is more competition between men and there are more offers for women.
When women propose and men are added to the market, the men cannot get better matches and the existing women cannot get worse matches since there is more competition between men, and women can make more offers.
The following result was already proven in \citep{gale1985ms}.
\begin{lemma}\label{lem:gsAddMen}
Consider a market $(M, W, P)$ and subsets $M_1 \subset M, W_1 \subset W$. Let $\mu = \GSalgShort(M, W, P)$, $\mu_1 = \GSalgShort(M_1, W, P|_{(M_1, W)})$ and $\mu_2 = \GSalgShort(M, W_1, P|_{(M, W_1)})$. 
\begin{itemize}
    \item Independently of whether men or women propose (but the same for $\mu_1$ and $\mu$), it holds that $\mu_1(m) \geq_m \mu(m) \; \forall m \in M_1, \mu_1(w) \leq_w \mu(w) \; \forall w \in W$.
    \item Independently of whether men or women propose (but the same for $\mu_2$ and $\mu$), it holds that $\mu_2(m) \leq_m \mu(m) \; \forall m \in M, \mu_2(w) \geq_w \mu(w) \; \forall w \in W_1$.
\end{itemize}
\end{lemma}
We apply this result to the $(A, S)$ and $(S, C)$ market. These results are independent of whether students or advisors or co-advisors propose. The proofs are deferred to Appendix~\ref{app:missingProofs}.
\begin{lemma}\label{lem:studentStaysMatchedAtNextIteration}
A student's match on the $(A, S)$ market can only improve over iterations (provided that he still participates).
More precisely, if student $s$ still participates in iteration $\tau$, $s \in S^{(\tau)}$, then $\mu_{AS}^{(\tau)}(s) \geq_s \mu_{AS}^{(\tau-1)}(s) \; \forall \tau \geq 2$.
This implies that $S_m^{(\tau-1)} \setminus S_u^{(\tau-1)} \subset S_m^{(\tau)} \; \forall \tau \geq 2$.
\end{lemma}
Intuitively, the remaining students face less competition from other students over iterations.
As we see now for the student-co-advisor market, although students are removed, the co-advisors' matches do not decrease since only ``bad'' students (that no co-advisor wants) are removed.
\begin{lemma}\label{lem:coadvisorCannotDecrease}
A co-advisor's match on the $(S, C)$ market can only improve over iterations: $\mu_{SC}^{(\tau)}(c) \geq_s \mu_{SC}^{(\tau-1)}(c) \; \forall \tau \geq 2, c \in C$.
\end{lemma}
In fact, Lemma~\ref{lem:coadvisorCannotDecrease} can make an additional statement as indicated in the proof. A student's match on the $(S, C)$ market cannot increase over iterations once he is part of the $(S, C)$ market: if student $s$ participates in iterations $\tau-1$ and $\tau$, i.e., $s \in S_m^{(\tau-1)} \cap S_m^{(\tau)} = S_m^{(\tau-1)} \setminus S_u^{(\tau-1)}$, then $\mu_{SC}^{(\tau)}(s) \leq_s \mu_{SC}^{(\tau-1)}(s)$.
So, even if we kept a student for some more iterations in the hope that he finds a co-advisor, he never will and instead unnecessarily hinders an advisor from matching with another student.
Therefore, it is safe to remove all students $S_u$ and not just a subset.
With this, we can finally prove stability.
\begin{theorem}
\label{thm:stablePhDAlgorithm}
The PhD algorithm returns a stable matching.
\end{theorem}%
\begin{proof}
By Proposition~\ref{prop:stableUnderStrictification}, it is sufficient to show that the matching $\mu$ is stable under the strictified preferences (first line of the algorithm).
We proceed by contradiction.
The PhD algorithm cannot return an individually irrational matching since the GS algorithms never return individually irrational matchings on the two-sided markets $(A, S)$ and $(S, C)$.
Assume that $(a, s, c) \in A \times S \times C$ blocks $\mu$. 
If student $s$ was not removed during the iterations of the PhD algorithm, this means $s \in S^{(T)}$, where $S^{(T)}$ is the set at the final iteration $T$.
$(a, s, c)$ must also block $\mu|_{(A, S^{(T)}, C)}$.
So either the matching $\mu|_{(A, S^{(T)})} = \mu_{AS}^{(T)}$ is blocked by $(a, s)$ or the matching $\mu|_{(S_m^{(T)}, C)} = \mu_{SC}^{(T)}$ is blocked by $(s, c)$. This contradicts the stability of the GS algorithms applied to the markets $(A, S^{(T)})$ and $(S_m^{(T)}, C)$.
If student $s$ was removed at some iteration $\tau$, $s$ could find an advisor, but no co-advisor at that iteration. Since $(a, s, c)$ is a blocking triple, $s$ and $c$ are mutually acceptable.
If students propose to co-advisors, $s$ must have proposed to $c$ in the GS algorithm. $c$ must have rejected $s$ in favor of another student $\tilde{s} >_c s$.
If co-advisors propose, $c$ must not have proposed to $s$ (since $s$ would not have rejected) and must therefore have matched with another student $\tilde{s} >_c s$.
In both cases, $\mu_{SC}^{(\tau)}(c) \geq_c \tilde{s} >_c s$ at that iteration.
By Lemma~\ref{lem:coadvisorCannotDecrease}, we must have $\mu(c)_S = \mu_{SC}^{(T)}(c) \geq_c \mu_{SC}^{(\tau)}(c) \geq_c \tilde{s} >_c s$. Hence $\mu(c)_S >_c s$ and $(a, s, c)$ cannot block $\mu$.
\end{proof}

We now discuss an important \emph{invariance property} of the PhD algorithm.
\begin{theorem}\label{thm:invariancePhDAlg}
Assume that we replace the GS algorithms with any (individually rational) stable matching algorithms, possibly differing across iterations and different on both submarkets.
Then, this variant of the PhD algorithm returns a stable matching.
Moreover, all such variants match the same set of advisors, students and co-advisors.
\end{theorem}
\begin{proofsketch}
By Proposition~\ref{prop:twoSidedSameMatched}, the set of matched persons is the same for any stable matchings on two-sided markets, so the sets $S_m^{(\tau)}$ and $S_u^{(\tau)}$ are the same at each iteration $\tau$ for all variants and they all run for the same number of iterations. Stability of the variant can be proven by comparing to the variant where students propose to co-advisors in the GS algorithm (which is known to be stable by Theorem~\ref{thm:stablePhDAlgorithm}).
\end{proofsketch}

This gives rise to a family of matchings produced by the PhD algorithm based on the two-sided matching algorithms. 
The stable matching mechanisms used at the final iteration do not affect stability and the set of matched people, but determine properties of the matching such as optimality.
In particular, the PhD algorithm returns a stable matching independently of which side proposes on the two-sided markets. If $S \implies A$ means that students propose to advisors, we have the variants: $A \implies S \implies C$, $A \impliedby S \implies C$, $A \impliedby S \impliedby C$ and $A \implies S \impliedby C$. These variants are optimal for some of the parties on (submarkets of) the associated two-sided markets.
Another variant is to swap the roles of advisors and co-advisors, i.e., first run the student-co-advisor market and then the student-advisor market on the matched students. 
The stability notion is the same (see Appendix~\ref{app:missingProofs}), but the matching returned by this variant is generally different.

\paragraph{Running time}
Since the GS algorithm runs in time $\mathcal{O}(|M| |W|)$ on the market $(M, W)$, a naive bound of the running time of the PhD algorithm is $\mathcal{O}((|A| + |C|) |S|^2)$.
Indeed, at least one student is removed at each iteration and $\mathcal{O}(\sum_{\tau=1}^{T} |A| |S^{(\tau)}| + |C| |S_m^{(\tau)}|) = \mathcal{O}((|A| + |C|) |S|^2)$.
At all iterations except the last, we only need to find a stable matching to determine the set of matched students.
A careful implementation can find such stable matchings in time $\mathcal{O}((|A|+|C| \log |S|) |S|)$ or $\mathcal{O}((|A|+|C| \log |C|) |S|)$ as we show now.
We modify the GS algorithm so that it takes the stable matching from the previous iteration as input and computes a new stable matching on a market where some students are added or removed. This need not correspond to a GS matching.
\begin{proposition}
\label{prop:gsPhDEfficientImplementation}
Modify the GS algorithm such that it takes the matching $\mu$ as input, each woman keeps track of the men she rejected in a waiting list L. Then we have:
\begin{itemize}
    \item if $\{ M^{(\tau)} \}_{\tau=1}^T$ is a decreasing sequence of men such that $M^{(\tau+1)} \subset M^{(\tau)}$ and $M^{(1)} = M$, then stable matchings $\{ \mu^{(\tau)} \}_{\tau = 1}^T$ such that $\mu^{(\tau)}$ is stable on $(M^{(\tau)}, W)$ can be output in $\mathcal{O}(|M| |W| \log |M| + T)$ total time in terms of their differences;
    \item if $\{ W^{(\tau)} \}_{\tau=1}^T$ is a decreasing sequence of women such that $W^{(\tau+1)} \subset W^{(\tau)}$ and $W^{(1)} = W$, then stable matchings $\{ \mu^{(\tau)} \}_{\tau = 1}^T$ such that $\mu^{(\tau)}$ is stable on $(M, W^{(\tau)})$ can be output in $\mathcal{O}(|M| |W| + T)$ total time in terms of their differences.
\end{itemize}
We output the first matching $\mu^{(1)}$. Matching $\mu^{(\tau+1)}$ is output in terms of its difference to the matching $\mu^{(\tau)}$.
(If the matching $\mu^{(\tau+1)}$ is output without reference to $\mu^{(\tau)}$, it takes $\mathcal{O}(|M| |W| + (|M|+|W|) T)$ total time.)
\end{proposition}
\begin{proofsketch}
Compute the initial matching $\mu^{(1)} = \mu$ in time $\mathcal{O}(|M| |W|)$. We will show that listing the consecutive differences between $\mu^{(\tau+1)}$ and $\mu^{(\tau)}$ takes at most $\mathcal{O}(|M| |W| + T)$ summed over all $\tau = 2, \dots, \tau$. Without loss of generality, we assume that a single man (woman) is removed between $M^{(\tau)}$ and $M^{(\tau+1)}$ ($W^{(\tau)}$ and $W^{(\tau+1)}$).
Suppose that man $m$ is removed. Remove $m$ from everyone's preferences.
If he was not matched to any woman, $\mu^{(\tau+1)}$ is the same as $\mu^{(\tau)}$, except that man $m$ is not listed. 
If he was matched to woman $w$, woman $w$ picks the most preferred man $\tilde{m}$ on her waiting list, removes him from this list and accepts $\tilde{m}$. 
If $\tilde{m}$ was matched to a woman, apply the same to $\tilde{m}$.
Since each woman lists at most $|M|$ men on her waiting list, the total running time is at most $\mathcal{O}(|M| |W| \log |W| + T)$ summed over all $\tau = 2, \dots, T$, if we implement the waiting lists with max-heaps.
Appendix~\ref{app:phdEfficientImpl} shows that the matching is stable.
Suppose that woman $w$ is removed. Remove $w$ from everyone's preferences. 
If she was not matched to any man, $\mu^{(\tau+1)}$ is the same as $\mu^{(\tau)}$, except that woman $w$ is not listed. 
If woman $w$ was matched to a man $m$, continue the GS algorithm starting from the matching $\mu^{(\tau)}$.
Since each man lists at most $|W|$ women, the total running time is at most $\mathcal{O}(|M| |W| + T)$ summed over all $\tau = 2, \dots, T$.
\end{proofsketch}
For the PhD algorithm, it is enough to keep track of the matching and the matched people on the $(A, S)$ and $(S, C)$ markets at each iteration. Therefore, if advisors propose to students, the total time spent is either $\mathcal{O}((|A| + |C| \log |S|) |S|)$ when students propose to co-advisors or $\mathcal{O}((|A| + |C| \log |C|) |S|)$ when co-advisors propose to students.
The space overhead for the waiting lists is bounded by the space needed to store the preference lists, i.e., $\mathcal{O}((|A| + |C|) |S|)$.
Since the matching returned by the PhD algorithm only depends on the matching algorithms used at the last iteration, we can use the GS algorithms for iterations $1, \dots, T-1$. Therefore, the PhD algorithm generally runs in time and space complexity $\mathcal{O}((|A| + |C|)|S|)$ plus the running time and space requirements of the matching algorithms at the final iteration.

\paragraph{Incompatible advisor pairs}
In practice, the same professor may act as both advisor and co-advisor and no student should have him as advisor and co-advisor simultaneously. Alternatively, it may be that the advisor and co-advisor in any matched triple should be from different countries to foster exchanges.
We present two modifications of the PhD algorithm to deal with this scenario. In Appendix~\ref{sec:incompatibleAdvisorPairs}, we show that they do not always result in stable matchings.
Consider an iteration of the PhD algorithm and a student who is matched to an advisor. Then, he only needs to consider co-advisors that are compatible with his (temporally) matched co-advisor. However, this can result in the situation when a very good student is matched to an advisor who is incompatible with all (or many) co-advisors. This means that this student is removed although he would find a match if he were matched to a different advisor in the first place. To accomodate for this, we propose to remove that advisor from the student's preferences and not remove the student.
This also fails when a student first matches with advisor $A_1$ and co-advisor $C_1$ such that $A_1$ is incompatible with some more preferred advisor $C_2$. In the next round, this student matches with a more preferred advisor $A_2$ and more preferred co-advisor $C_2$. $A_2$ may be available because another student no longer matched to him. This means that co-advisor $C_1$ is unmatched at the second round, although some other student may be available.
In practice, this heuristic may work well however.
Note that Proposition~\ref{prop:gsPhDEfficientImplementation} can be adapted and the overall complexity of this algorithm is still $\mathcal{O}((|A| + |C|) |S|)$ (up to log factors).
The authors of \citep{zhong2019Cooper} claim that they can find a stable matching in this setting. We provide a counterexample when their algorithm does not work, even when everyone has complete preferences.
Without incompatible advisor pairs, we further show that their algorithm is a special case of our PhD algorithm with advisors proposing to students and students proposing to co-advisors. The advantage of our algorithm is that we can use any stable matching algorithms.
We have reasons to believe that finding a stable matching (if it exists) with incompatible advisor pairs and consistent preferences may be NP-complete.

\paragraph{Arbitrary quotas}
We now discuss the extension when persons can have quota greater than one. This setting is given in more detail in Appendix~\ref{app:quotas}.
For example, each student needs to complete a maximum number of projects and each professor can advise and/or co-advise a maximum number of students/projects.
Let us first consider two-sided markets. If men and women have quotas, one can pass through the extended market construction where each person is replicated according to their quota with identical preferences \citep{roth1984evolution}. Because the relationships between replicas of the same person are lost, this also allows the same man and woman to be matched more than once. If this is undesirable, one can modify the GS algorithm such that a man proposes to his most preferred woman he did not propose to already whenever he has spots free, or to himself if no acceptable women are left. A woman accepts whenever she prefers the man to any of her current match partners. 
If the same man and woman are allowed to match more than once, a man instead proposes to his most preferred woman who did not reject him yet.
For the following, it does not matter whether the same man and woman are allowed to match more than once or not.
The stability and optimality proofs then carry over from the one-to-one to this many-to-many setting. Alternatively, one can view it as a special case of the results \citep{roth1991natural, echenique2004theory} with strongly substitutable preferences.\\
We can treat a three-sided market as two two-sided markets with the constraint that the same student must have the same number of match partners on both two-sided markets.%
\footnote{With this treatment, a student is matched to a set of advisors $\mathcal{A}$ and co-advisors $\mathcal{C}$ such that $|\mathcal{A}| = |\mathcal{C}|$. The matched triples are not defined uniquely. For example, if $\mathcal{A} = \{a_1, a_2 \}, \mathcal{C} = \{ c_1, c_2 \}$, the matchings $\{ (a_1, s, c_1), (a_2, s, c_2) \}$ and $\{ (a_1, s, c_2), (a_2, s, c_1) \}$ may both be stable. Then, incompatibilities between advisor pairs (Section~\ref{sec:incompatibleAdvisorPairs}) could rule out one of these options.}
Depending on the application, the same advisor-student pair can match more than once or not, and the same for the student-co-advisor pair.
In the PhD algorithm without quotas, a student is removed if he does not find a co-advisor. Here, when a student has a quota on the $(S, C)$ market and does not use all of his quota, we reduce his quota on the $(A, S)$ market by the number of unfilled spots.
More precisely, say a student matches with $k_a$ advisors at some iteration. Then, he is assigned quota $k_a$ on the student-co-advisor market. If this student then matches to $k_c \leq k_a$ co-advisors, reduce his capacity on the advisor market by $k_a - k_c$ (instead of removing the student).
Lemmata \ref{lem:gsAddMen}, \ref{lem:studentStaysMatchedAtNextIteration}, \ref{lem:coadvisorCannotDecrease} and Theorems \ref{thm:stablePhDAlgorithm}, \ref{thm:invariancePhDAlg} can be adapted to prove stability. We refer to Appendix~\ref{app:quotas}.

\section{Multi-sided markets}\label{sec:nsidedMarketsShort}
We extend 3-sided markets to $n$-sided markets for $n \geq 2$ and generalize the PhD algorithm to obtain Algorithm~\ref{alg:phdAlgorithmNSided}.
This model is appropriate when $n$ sides must be matched, with a start side and an end side, and when the $i$-th side only cares about the $(i-1)$-th and $(i+1)$-th side separately.
Let $S_1, \dots, S_n$ the $n \geq 2$ sides of the market. The $n$-sided matching market with preferences $P$ is described by $((S_1, \dots, S_n), P)$.
In general, preferences $P$ are such that $P(s_k)$ are the preferences of $s_k \in S_k$ and can be described by a total order on $S_1 \times \dots \times S_{k-1} \times S_{k+1} \times \dots \times S_n \cup \{ (s_k, \dots, s_k) \}$, where $(s_k, \dots, s_k) \in S_k^{n-1}$ means that $s_k$ stays single.
We consider markets such that $s_k \in S_k$ has separate preferences $P_{k-1 \leftarrow k}$ over $S_{k-1} \cup \{ s_k \}$ (for $k > 1$) and $P_{k \rightarrow k+1}$ over $S_{k+1} \cup \{ s_k \}$ (for $k < n$), and is indifferent between participants from the other sides. Call market $k$ the two-sided market $(S_k, S_{k+1}, P_k)$, where $P_k = \{ P_{k \leftarrow k+1}, P_{k \rightarrow k+1} \}$ are the two-sided preferences between $S_k$ and $S_{k+1}$.
The PhD market in our setting is hence specified by $((S_1, \dots, S_n), (P_1, \dots, P_{n-1}))$. When we write $>_{s_k}$, it is implicit whether it refers to $P_{k-1 \leftarrow k}(s_k)$ or $P_{k \rightarrow k+1}(s_k)$.

\begin{definition}
A matching is a set $\mu \subset S_1 \times \dots \times S_n$ such that each person is matched at most once.
For $s_k \in S_k$, define the match partner of $s_k$ as $\mu(s_k) = (s_1, \dots, s_{k-1}, s_{k+1}, \dots, s_n)$ if $(s_1, \dots, s_{k-1}, s_k, s_{k+1}, \dots, s_n) \in \mu$. If $s_k$ has no match partner, let $\mu(p) = (p, \dots, p) = \{ p \}^{n-1}$.
\end{definition}
Let $\mu_k$ the matching restricted to the market $(S_k, S_{k+1}, P_k)$, i.e., $\mu_k = \{ (s_k, s_{k+1}) \mid (s_1, \dots, s_n) \in \mu \} \subset S_k \times S_{k+1}$.
The matching $\mu$ is \emph{individually rational} if it is individually rational on the two-sided markets. We define stability in terms of two-sided stability.
\begin{definition}
The matching is stable if it is individually rational and there exists no blocking $n$-tuple $(s_1, \dots, s_n) \in S_1 \times \dots \times S_n \setminus \mu$ such that for all $k = 1, \dots, n-1$ we have either $(s_k, s_{k+1}) \in \mu_k$ or $(s_k, s_{k+1})$ blocks $ \mu_k$.
\end{definition}
In other words, for all $k = 1, \dots, n-1$, both persons in the couple $(s_k, s_{k+1})$ improve their match or, if they are matched, they can keep their match, i.e., $s_{k} = \mu_{k}(s_{k+1})$. At least one $(s_k, s_{k+1})$ has to improve their match for some $k$. The following theorem is proved in Appendix~\ref{sec:nsidedMarkets}.
\begin{theorem}
\label{thm:stablePhDAlgorithmNSided}
The $n$-sided PhD Algorithm returns a stable matching.
\end{theorem}
This setting could model (multi-source) supply chains such as car or vaccine manufacturing \citep{world2021covid}.
To produce a vaccine, one first needs to buy the chemical reactants, then they must be mixed together, packaged and finally shipped (which itself involves a shipping supply chain). All these $n$ steps must be executed in order, and factors such as trade restrictions, distances between production plants inform the preferences that each side has over the previous side and the next side. 
We assume that each side only has preferences over the previous and next side. This should hold as long as each step is sufficiently self-contained/standardized, e.g., the packaging procedure works independently of the chemicals used to produce the vaccine.
Half-produced vaccines are undesirable, so only full matches are of interest.
If a match is unstable, some entity in the supply chain has an incentive to deviate, so we search for stable matchings.

\section{Numerical simulation}
Since real preference data is difficult to obtain, we generate an artificial dataset of a PhD market, where we are interested in full matches only. Each person is randomly assigned a random number of research fields. An advisor computes the overlap in research fields he has with students to order them. We do this for students and co-advisors as well. 
When a professor and a student have high research overlap, they are likely to both rank each other highly, so we add random jitter to make preferences less structured. More details are in Appendix~\ref{app:numericalSimulation}.
As a baseline, we use the algorithm proposed by~\cite{danilov2003existence} which is only guaranteed to produce stable matchings when all pairs are mutually acceptable. This corresponds to stopping the PhD algorithm after one iteration.
We plot the number of (complete) matches and the number of blocking triples over iterations of the PhD algorithm in Figure~\ref{fig:numberMatchesBlockingTriplesSyntheticData}. The fact that the number of matches cannot go down is a consequence of Lemma~\ref{lem:coadvisorCannotDecrease}. 
Note that this can be observed for individual runs, but not necessarily in the above figure because we average over datasets that reach a given iteration, and the number of iterations depends on the dataset.
The number of blocking triples is also mostly monotonically decreasing, and it reaches zero at the last iteration as expected.

\begin{figure}
\centering
\begin{subfigure}{.5\textwidth}
  \centering
  \includegraphics[width=\textwidth,height=.16\textheight,keepaspectratio]{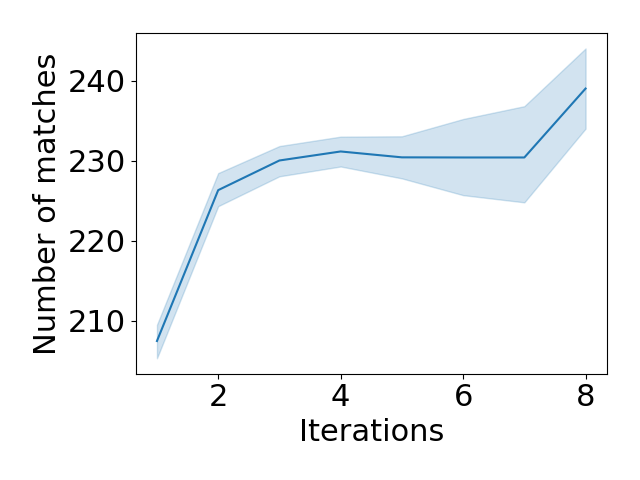}
  \label{fig:sub1}
\end{subfigure}%
\begin{subfigure}{.5\textwidth}
  \centering
  \includegraphics[width=\textwidth,height=.16\textheight,keepaspectratio]{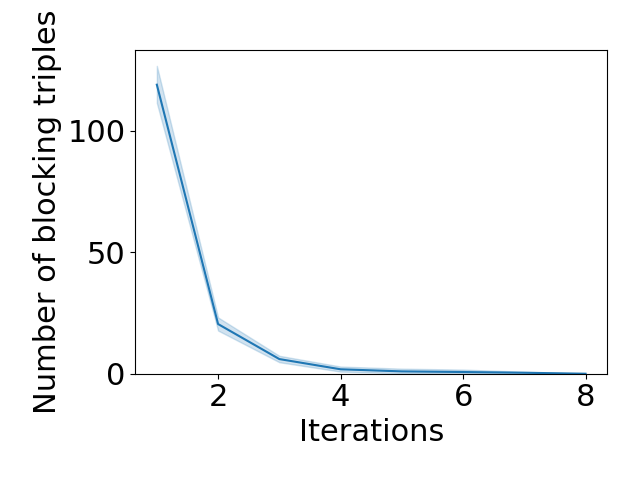}
  \label{fig:sub2}
\end{subfigure}
\caption{Number of complete matches and blocking triples of the matching produced over iterations of the PhD algorithm, averaged over 40 synthetic datasets. }
\label{fig:numberMatchesBlockingTriplesSyntheticData}
\end{figure}


\section{Conclusion}\label{sec:conclusion}
We have shown that stable matchings exist and can be found in (almost) quadratic time in three-sided markets for matching students with advisors and co-advisors under suitably defined preferences.
On a synthetic dataset, we have seen that the iterative nature increases the number of matches.
It would be interesting to obtain real supply chain data and quantify the benefits of using the PhD algorithm as compared to a naive baseline.
Our approach assumes that all persons know their (true) preferences. If the market is large, persons may not know each other.
So the PhD algorithm can be combined with preference modelling techniques, such as preference elicitation \citep{charlin2013toronto} or handcrafted criteria \citep{mordig2021two}. These techniques make use of preliminary matchings (based on tentative preferences) to give people time to form their preferences.

A possible future direction is to prove the NP-completeness when advisor pairs can be incompatible, similarly to the line of work in \citep{ng1991three, lam2019existence, huang2007two}.
With quotas, we wonder if stable matchings exist when the same triple rather than the same pairs can match at most once (as an extension to Section~\ref{sec:phdAlgorithm}), and how to find them.
Note that the PhD algorithm is not strategy-proof, i.e., persons can misrepresent their preferences to obtain more satisfying matches. An interesting question is how much information a person needs to successfully manipulate the outcome (see \citep{roth1989two} for two-sided markets with incomplete information).

\begin{ack}
This work was supported by the German Federal Ministry of Education and Research (BMBF): Tübingen AI Center, FKZ: 01IS18039B, and by the Machine Learning Cluster of Excellence, EXC number 2064/1 – Project number 390727645.
\end{ack}


\bibliographystyle{plainnat}
\bibliography{bibliography}

\clearpage

\appendix 

\section{Missing proofs}\label{app:missingProofs}
Based on man-optimality and woman-pessimality of the matching returned by the GS algorithm, we can show that the set of matched persons is the same in all stable matchings, as already stated in Section~\ref{sec:twosidedMarketsShort}. An alternative proof can be found in \citep{knuth1997stable}.
\begin{customprop}{\ref{prop:twoSidedSameMatched}}
Let ($M$, $W$, $P$) a two-sided market with strict preferences. The set of matched persons is the same in any stable matching.
\end{customprop}
\begin{proof}
Denote the matched men in an arbitrary matching $\mu$ by $R(\mu)$.
Let $\mu$ the GS matching (with men proposing) and $\lambda$ an arbitrary stable matching. It is sufficient to show that the matched persons are the same in $\mu$ and $\lambda$. By man-optimality, $\mu(m) \geq_m \lambda(m) \; \forall m \in M$. Whenever $m$ is matched in $\lambda$, we have $\mu(m) \geq_m \lambda(m) >_m m$, hence $\mu(m) >_m m$, so $m$ is also matched in $\mu$. In other words, $R(\lambda) \subset R(\mu)$. In particular, the number of matched men in $\mu$ is greater or equal to the number of matched men in $\lambda$. From woman-pessimality, we have that the number of matched women in $\lambda$ must be greater or equal to the number of matched women in $\mu$. Since the number of matched men equals the number of matched women for a given matching, we have $| R(\lambda) | \geq | R(\mu) |$, hence $R(\lambda) = R(\mu)$.
We can apply a similar reasoning to show that the matched women are the same in all matchings.
\end{proof}

We prove a theorem on the behavior of the GS algorithm when men and women are removed. See \citep{gale1985ms} for an alternative proof.
\begin{customlemma}{\ref{lem:gsAddMen}}
Consider a market $(M, W, P)$ with strict preferences $P$ and subsets $M_1 \subset M, W_1 \subset W$. Let $\mu = \GSalgShort(M, W, P)$, $\mu_1 = \GSalgShort(M_1, W, P|_{(M_1, W)})$ and $\mu_2 = \GSalgShort(M, W_1, P|_{(M, W_1)})$. 
\begin{itemize}
    \item Independently of whether men or women propose (but the same for $\mu_1$ and $\mu$), it holds that $\mu_1(m) \geq_m \mu(m) \; \forall m \in M_1, \mu_1(w) \leq_w \mu(w) \; \forall w \in W$.
    \item Independently of whether men or women propose (but the same for $\mu_2$ and $\mu$), it holds that $\mu_2(m) \leq_m \mu(m) \; \forall m \in M, \mu_2(w) \geq_w \mu(w) \; \forall w \in W_1$.
\end{itemize}
\end{customlemma}
\begin{proof}
By the man-optimality of the GS algorithm in the many-to-many setting (with strict preferences), the order in which free men propose does not matter. So we can choose the order in a suitable fashion. We first prove $\mu_1(m) \geq_m \mu(m) \; \forall m \in M_1, \mu_1(w) \leq_w \mu(w) \; \forall w \in W$ when men propose.
Consider the execution of $\GSalgShort(M, W, P)$ and first let only men $M_1$ propose in the GS algorithm until it terminates and returns $\mu_1$. Continue the GS algorithm and let all men propose, including men $M_1$ if they become unmatched, until it returns $\mu_2$. The inequalities follow since men $M_1$ cannot get better matches when they continue proposing to less preferred women and women only accept better men.\\
We next prove $\mu_2(m) \leq_m \mu(m) \; \forall m \in M, \mu_2(w) \geq_w \mu(w) \; \forall w \in W_1$ when men propose.
Consider the execution of $\mu_2 = \GSalgShort(M, W_1, P|_{(M, W_1)})$. This is the same as running $\GSalgShort(M, W, P)$ where in addition a man is rejected whenever he proposes to a woman in $W \setminus W_1$. Whenever a man is rejected by a woman $w \in W_1$ in the execution of $\mu$, he is also rejected by a woman in the execution of $\mu_2$, hence $\mu_2(m) \leq_m \mu(m) \; \forall m \in M$. This also means that each man proposes to at least as many women in $\mu_2$ as in $\mu$ and the women's matches can only become better, $\mu_2(w) \geq_w \mu(w) \; \forall w \in W_1$.\\
When women propose, it follows from symmetry.
\end{proof}

We now provide the missing proofs for Section~\ref{sec:phdAlgorithm}.
\begin{customlemma}{\ref{lem:studentStaysMatchedAtNextIteration}}
A student's match on the $(A, S)$ market can only improve over iterations (provided that he still participates).
More precisely, if student $s$ still participates in iteration $\tau$, $s \in S^{(\tau)}$, then $\mu_{AS}^{(\tau)}(s) \geq_s \mu_{AS}^{(\tau-1)}(s) \; \forall \tau \geq 2$.
This implies that $S_m^{(\tau-1)} \setminus S_u^{(\tau-1)} \subset S_m^{(\tau)} \; \forall \tau \geq 2$.
\end{customlemma}
\begin{proof}
Since the number of students decreases, Lemma~\ref{lem:gsAddMen} applies with $M_1 = S^{(\tau)} \subset S^{(\tau-1)} = M$.
$s \in S_m^{(\tau-1)} \setminus S_u^{(\tau-1)}$ means that $\mu_{AS}^{(\tau-1)}(s) >_s s$ and $s \in S^{(\tau)}$. Therefore, $\mu_{AS}^{(\tau)}(s) >_s s$, i.e., $s \in S_m^{(\tau)}$.
\end{proof}

\begin{customlemma}{\ref{lem:coadvisorCannotDecrease}}
A co-advisor's match on the $(S, C)$ market can only improve over iterations: $\mu_{SC}^{(\tau)}(c) \geq_s \mu_{SC}^{(\tau-1)}(c) \; \forall \tau \geq 2, c \in C$.
\end{customlemma}
\begin{proof}
Consider iteration $\tau$ and the $(S, C)$ market where students $S_m^{(\tau)}$ participate.
First let only students $S_m^{(\tau-1)} \cap S_m^{(\tau)}$ participate to obtain the temporary matching $\lambda_{SC}^{(\tau)}$. 
By Lemma~\ref{lem:studentStaysMatchedAtNextIteration}, $S_m^{(\tau-1)} \cap S_m^{(\tau)} = ((S_m^{(\tau-1)} \setminus S_u^{(\tau-1)}) \cap S_m^{(\tau)}) \cup ((S_m^{(\tau-1)} \cap S_u^{(\tau-1)}) \cap S_m^{(\tau)}) = S_m^{(\tau-1)} \setminus S_u^{(\tau-1)}$.
Since the students $S_u^{(\tau-1)}$ are not matched to any co-advisors, $\lambda_{SC}^{(\tau)} = \mu_{SC}^{(\tau-1)}$ on the market restricted to $(S_m^{(\tau-1)} \setminus S_u^{(\tau-1)}, C)$.
Let the remaining students $S_m^{(\tau)} \setminus S_m^{(\tau-1)}$ participate as well.
By Lemma~\ref{lem:gsAddMen}, since more students have joined the market, we must have $\mu_{SC}^{(\tau)}(s) \leq_s \lambda_{SC}^{(\tau)}(s) = \mu_{SC}^{(\tau-1)}(s) \; \forall s \in S_m^{(\tau-1)} \setminus S_u^{(\tau-1)}$. Similarly, we obtain $\mu_{SC}^{(\tau)}(c) \geq_c \mu_{SC}^{(\tau-1)}(c) \; \forall c \in C$.
\end{proof}

\begin{customthm}{\ref{thm:invariancePhDAlg}}
Assume that we replace the GS algorithms with any (individually rational) stable matching algorithms, possibly differing across iterations and different on both submarkets.
Then, this variant of the PhD algorithm returns a stable matching.
Moreover, all such variants match the same set of advisors, students and co-advisors.
\end{customthm}
\begin{proof}
Call \emph{variant} a variant of the algorithm with a specific choice of the stable matching algorithms.
We compare such a variant (Variant 1) to the variant where the GS algorithm is run with students proposing to co-advisors (Variant 2).
Fix an iteration $\tau$ of the PhD algorithm with students $S^{(\tau)}$ remaining.
By Proposition~\ref{prop:twoSidedSameMatched}, the set of matched persons is the same for any stable matchings on two-sided markets, so the set $S_m^{(\tau)}$ is independent of which stable matching algorithm is used on the market $(A, S^{(\tau)})$ at that iteration. Given $S_m^{(\tau)}$, the set of unmatched students $S_u^{(\tau)}$ is independent of the stable matching algorithm that is used on the $(S_m^{(\tau)}, C)$ market.
Therefore, both variants run for the same number of iterations and match the same students.\\
So the matching produced by the PhD algorithm only depends on the matching algorithms that are used at the last iteration.
Suppose Variant 1 returns matching $\lambda$ and Variant 2 returns matching $\mu$. Let $(a, s, c)$ be a blocking triple. We repeat the proof of Theorem~\ref{thm:stablePhDAlgorithm}.
If $s$ was not removed, we have the same conclusion because the matching mechanisms are stable. If $s$ was removed, we know from co-advisor pessimality of $\mu$     that
$\lambda(c)_S = \lambda_{SC}^{(T)}(c) \geq_c \mu_{SC}^{(T)}(c) >_c s$. Hence $\lambda(c)_S >_c s$ and $(a, s, c)$ cannot block $\lambda$.
\end{proof}

\paragraph{Nonstrict preferences}
In reality, people often have non-strict preferences. Here, we show that we can strictify the preferences and run the algorithms on the strict preferences. The resulting matching is then also stable under the original non-strict preferences.
Since the three-sided markets we consider are made up of two-sided markets, we only discuss non-strict preferences for two-sided markets.
A pre-order $\leq$ is a partial order, except that the implication $x \leq y \land y \leq x \implies x = y$ is not required. For person $p$ with total pre-order $\leq_p$, we define indifference $=_p$ by $p_1 =_p p_2 \iff p_1 \leq_p p_2 \land p_2 \leq_p p_1$. In words, $p$ is indifferent between $p_1$ and $p_2$. We also define strict preference $<_p$ by $p_1 <_p p_2 \iff p_1 \leq_p p_2 \land p_2 \nleq_p p_1$.
For the two-sided market, we assume that man $m$ has preferences $P(m)$, which is a total pre-order over $W \cup \{ m \}$ denoted by $\leq_m$, and woman $w$ has preferences $P(w)$, which is a total pre-order over $M \cup \{ w \}$ denoted by $\leq_w$.
We assume that a person is never indifferent between himself and anyone else, i.e., $w \neq_m m, m \neq_w w \; \forall m \in M, w \in W$.
Stability is defined as before in terms of strict inequalities $<_p$ rather than $\leq_p$.
Preferences are strictified by breaking some (or all) ties:
\begin{definition}
The total pre-order $\geq_{\tilde{P}}$ is a strictification of the total pre-order $\geq_P$ if $p_1 >_p^{P} p_2 \implies p_1 >_p^{\tilde{P}} p_2$ for person $p$ and $p_1, p_2$ are persons for which $>_p^P$ is defined. 
The superscript refers to the preferences which are used. 
Equivalently, $p_1 \geq_p^{\tilde{P}} p_2 \implies p_1 \geq_p^{P} p_2$.
\end{definition}

\begin{proposition}
\label{prop:stableUnderStrictification}
Consider a two-sided and let $\tilde{P}$ be a strictification of $P$ (for all persons). If the matching $\mu$ is stable under $\tilde{P}$, it is stable under $P$.
\end{proposition}
\begin{proof}
By contradiction, assume that $\mu$ is not stable under $P$. Then, there must exist a blocking pair $(m, w)$ of $\mu$ such that $m >_w^{P} \mu(w)$ and $w >_m^{P} \mu(m)$ ($m = w$ also possible). This implies $m >_w^{\tilde{P}} \mu(w)$ and $w >_m^{\tilde{P}} \mu(m)$. Therefore, $(m, w)$ blocks $\mu$ under $\tilde{P}$ and $\mu$ is not stable under $\tilde{P}$.
\end{proof}

The GS algorithm deals with non-strict preferences by breaking ties before it is called. Since the tie breaking is arbitrary, it need not return a unique matching.
A similar result to Proposition~\ref{prop:stableUnderStrictification} holds for three-sided markets if the preferences on each two-sided market are strictified (in our setting of consistent preferences of students).
For the PhD algorithm, it is in fact necessary to strictify the preferences before it is called rather than at each iteration. Otherwise, the returned matching may not be stable.

\paragraph{Advisor-co-advisor symmetry}
We motivate our definition of three-sided stability given in Section~\ref{sec:threeSidedMarketsShort}, and we show that our definition of stability coincides with the one given in \cite{zhang2019Hybrid, huang2007two, manlove2017almost}, which treat advisors and co-advisors symmetrically. In fact, these definitions show that advisors and co-advisors could also be treated on an equal footing, unlike what one expects from our definition $(a, c) >_s (\tilde{a}, \tilde{c}) \iff a >_s^A \tilde{a} \lor (a = \tilde{a} \land c >_s^C \tilde{c})$ where students favor advisors over co-advisors.
By analogy with the two-sided stability definition, where we additionally allow only the advisor-student or the student-co-advisor side to change, we say that the mutually acceptable $(a, s, c)$ is a blocking triple if:
\begin{align*}
\begin{cases}
(a, c) >_s \mu(s)
\iff
a >_s \mu(s)_A \lor (a = \mu(s)_A \land c >_s \mu(s)_C) \\
s >_a \mu(a)_S \lor s =  \mu(a)_S \\
s >_c \mu(c)_S \lor s = \mu(c)_S.
\end{cases}.
\end{align*}
The previous can be reformulated as:
\begin{align*}
\begin{cases}
a >_s \mu(s)_A \\
s >_a \mu(a)_S \\
s >_c \mu(c)_S \\
c >_s \mu(s)_C \text{ (added)}
\end{cases}
\lor
\begin{cases}
c >_s \mu(s)_C \\
s = \mu(a)_S \\
s >_c \mu(c)_S
\end{cases}
\lor
\begin{cases}
a >_s \mu(s)_A \\
s >_a \mu(a)_S \\
s = \mu(c)_S
\end{cases}.
\end{align*}
We have added the line in the first clause as highlighted since it does not change the definition of stability.
If not, since $(s, c)$ are required to be mutually acceptable, $s <_s c \leq_s \mu(s)_C$ implies that $s$ is matched and is hence covered by the third clause already for $(a, s, \mu(s)_C)$. 
This corresponds to the description of stability in Section~\ref{sec:threeSidedMarketsShort}.
We can also arrive at this stability definition if we define the student's preferences symmetrically:
\begin{equation*} 
\begin{aligned}
(a, c) >_s (\tilde{a}, \tilde{c})
&\iff 
(a >_s \tilde{a} \land (c >_s \tilde{c} \lor c = \tilde{c})) \lor ((a >_s \tilde{a} \lor a = \tilde{a}) \land c >_s \tilde{c}).
\end{aligned}
\end{equation*}
The same matchings are stable, whether we define $>_s$ to favor advisors over co-advisors, or to treat both symmetrically.

\subsection{Efficient implementation of the PhD algorithm}\label{app:phdEfficientImpl}
todo: remove this sentence:
(See Proposition~\ref{prop:gsPhDEfficientImplementation})
Proposition~\ref{prop:gsPhDEfficientImplementation} misses some log factors which are corrected below.

\newcommand{\MFPAlg}{\textsc{MFP}}
\newcommand{\extractFromWlist}{\textsc{extractFromWlist}}
\newcommand{\addToWlist}{\textsc{addToWlist}}

To implement the PhD algorithm, we need to efficiently compute the matching after men and women join or leave the market.
For this, we modify the GS algorithm so that it starts from a given matching and keeps track of waiting lists.
Fix the full market $(M_f, W_f)$ that contains all persons. We consider a submarket $(M, W) \subset (M_f, W_f), M \subset M_f, W \subset W_f$ that consists of the active persons $M \cup W$ that we want to match.
Whenever a man proposes to a woman and is rejected, the woman puts the man on her waiting list. In the event that her matched man leaves the market later on, she can propose to the most preferred man on her waiting list to try to get a new match.
Note that a man only proposes if he is part of the active men $M$ ($m \in M$). A woman only accepts a man if he is better than her current match and she is active ($w \in W$). Otherwise, she puts the man on her waiting list.
We call the modified algorithm the \MFPAlg~algorithm (abbreviation for MatchFreePersons) which is shown in Algorithm~\ref{alg:mfpAlg}.
The function $\extractFromWlist(L(p))$ extracts and removes the most preferred person from the waiting list $L(p)$. If $L(p)$ is empty, it returns $p$ himself. The function $\addToWlist(L(p), \tilde{p})$ adds person $\tilde{p}$ to $p$'s waiting list $L(p)$.
\begin{algorithm}[ht]
\DontPrintSemicolon
\SetAlgoLined
\KwIn{marriage market $(M, W)$, matching $\mu$, waiting lists $L = \{ L(p) \}_{p \in M_f \cup W_f}$; see the text for the requirements on the input
}
\KwOut{matching $\mu$, waiting lists $L$}
\While{there is an unmatched man $m \in M$, i.e., $\mu(m) = \emptyset$}{
 \tcp{\footnotesize only active men propose}
 $w \gets \extractFromWlist(L(m))$ \;
 \uIf{$w = m$}{
   $\mu \gets \mu \cup \{ (m, m) \}$
 }
 \uElseIf{$m >_w \mu(w)$ and $w \in W$}{ \tcp{\footnotesize woman always rejects if she is not active}
    $\mu \gets \mu \setminus \{ (\mu(w), w) \} \cup \{ (m, w) \}$ \;
    \If{$\mu(w) \neq w$}{
      $\addToWlist(L(w), \mu(w))$
    }
 }
 \Else{
    \If{$m$ is acceptable to $w$}{
      $\addToWlist(L(w), m)$
    }
 }
}
\Return{$\mu$}
\caption{Marriage Market - MFP algorithm inspired by the GS algorithm \citep{gale1962college}
}
\label{alg:mfpAlg}
\end{algorithm}
Let $R(\mu) = \{ p \in M \cup W \mid \mu(p) \neq \emptyset \}$ the set of matched persons, including self-matches.
We require that the inputted matching $\mu$ (defined on the market $(M, W)$ and waiting lists $L$ satisfy the following:
\begin{itemize}
    \item The matching $\mu|_{R(\mu)}$ restricted to the market of matched persons in $\mu$ is stable and matches all women, i.e., $W \subset R(\mu) \subset M \cup W$.
    \item The waiting lists are compatible with matching $\mu$:
    \begin{itemize}
        \item The waiting list $L(p)$ of person $p \in M_f \cup W_f$ (in the full market) contains all matched active persons $\tilde{p} \in M \cup W$ such that $\tilde{p}$ prefers $p$ to his current match $\mu(\tilde{p})$. In other words, we require $(p >_{\tilde{p}} \mu(\tilde{p})) \land (\mu(\tilde{p}) \neq \emptyset) \implies \tilde{p} \in L(p) \; \forall p \in M_f \cup W_f, \forall \tilde{p} \in M \cup W$.
        \item Whenever $(m, w) \in M_f \times W_f$ are mutually acceptable, they are matched or one appears in the waiting list of the other. So we require $\mu(m) = w \lor m \in L(w) \lor w \in L(m) \; \forall m \in M_f, w \in W_f$.
    \end{itemize}
\end{itemize}
\begin{proposition}
If the matching $\mu_0$ and waiting lists $L_0$ satisfy these requirements, the outputted matching and waiting lists $\mu, L = \MFPAlg(M, W, \mu_0, L_0)$ also satisfy these requirements.
\end{proposition}
\begin{proof}
Stability of $\mu$ can be proved along the lines of the proof of the traditional GS algorithm.
Since men only propose to acceptable women and women only accept more preferred men, $\mu(m) \geq_m m$ and $\mu(w) \geq_w \mu_0(w) \geq_w w$, the produced matching must be individually rational.
Suppose $(m, w)$ blocks $\mu$. Then, $w >_m \mu(m)$. If $w$ was on $m$'s waiting list, $m$ must have proposed to $w$, so $\mu(w) \geq_w m$. Hence $(m, w)$ cannot be a blocking pair.
If $w$ was not on $m$'s waiting list, this means $\mu_0(w) >_w m$ because the waiting list $L_0(m)$ is compatible with $\mu_0$. Hence $\mu(w) \geq_w \mu_0(w) >_w m$, so $(m, w)$ cannot be a blocking pair.
We now prove that the returned waiting lists are compatible. Note that $\mu(p) \neq \emptyset \; \forall p \in M \cup W$ once the algorithm terminates.
Consider the waiting list of man $m \in M_f$ and suppose that there exists a woman $w$ who is not on $m$'s waiting list such that $m >_w \mu(w)$. Hence $m >_w \mu_0(w)$, so $w$ was on $m$'s waiting list $L_0(m)$ by compatibility. This means that $m$ must have proposed to $w$, so $\mu(w) \geq_w m$, a contradiction.
Consider the waiting list of woman $w \in W_f$ and suppose that there exists a man $m \in M$ such that $w >_m \mu(m)$. If he is not on $w$'s waiting list $L(w)$, he was also not on $w$'s initial waiting list $L_0(w)$, so $w$ must have been on $m$'s waiting list $L_0(m)$ (by compatibility). Since he is not on $w$'s waiting list $L(w)$, this means that he did not propose to $w$ and is matched to $\mu(m) >_m w$, a contradiction.
Finally, the algorithm preserves the property $\mu(m) = w \lor m \in L(w) \lor w \in L(m) \; \forall m \in M_f, w \in W_f$. Indeed, when $w \neq m$ is removed from $m$'s waiting list, either $(m, w)$ are matched or $w$ rejects $m$ and puts $m$ on her waiting list.
\end{proof}

Given a stable matching $\mu$ on the market $(M, W)$ with waiting lists $L$, we can do the following:
\begin{itemize}
\item To add a man $m \in M_f \setminus M$, add the man to the market and run the algorithm: $\MFPAlg(M \cup \{ m \}, W, \tilde{\mu}, L)$, where $\tilde{\mu}(m) = \emptyset, \tilde{\mu}(p) = \mu(p) \; \forall p \in M \cup W$.
\item To add a woman $w \in W_f \setminus W$, add the woman to the market and run the algorithm with women proposing: $\MFPAlg(W \cup \{ w \}, M, \tilde{\mu}, L)$, where $\tilde{\mu}(w) = \emptyset, \tilde{\mu}(p) = \mu(p) \; \forall p \in M \cup W$.
\item To remove a man $m \in M$, put his matched woman $w$ on his waiting list (if any), unmatch $w$, remove the man from the market and run the algorithm with women proposing: $\MFPAlg(W, M \setminus \{ m \}, \tilde{\mu}, \tilde{L})$, where $\tilde{L}(m) = L(m) \cup \{ \mu(m) \}, \tilde{\mu}(\mu(m)) = \emptyset$ (if $\mu(m) \neq m$) and $\tilde{L}(p) = L(p), \tilde{\mu}(p) = \mu(p)$ otherwise, where $\tilde{\mu}$ is defined over $(W, M \setminus \{ m \})$.
\item To remove a woman $w \in W$, put her matched man $m$ on her waiting list (if any), unmatch $m$, remove the woman from the market and run the algorithm with men proposing: $\MFPAlg(M, W \setminus \{ w \}, \tilde{\mu}, \tilde{L})$, where $\tilde{L}(w) = L(w) \cup \{ \mu(w) \}, \tilde{\mu}(\mu(w)) = \emptyset$ (if $\mu(w) \neq w$) and $\tilde{L}(p) = L(p), \tilde{\mu}(p) = \mu(p)$ otherwise, where $\tilde{\mu}$ is defined over $(M, W \setminus \{ w \})$.
\end{itemize}
Provided that the input matching $\mu$ is stable and the waiting lists are compatible with $\mu$ on the market $(M, W)$:
\begin{itemize}
\item We add a man $m$ by calling $\MFPAlg(M \cup \{ m \}, W, \tilde{\mu}, L)$.
The only free person is $m$, so $R(\tilde{\mu}) = R(\mu)$ and $\tilde{\mu}|_{R(\tilde{\mu})} = \mu|_{R(\mu)}$ is stable. The waiting lists are compatible with $\tilde{\mu}$ since $\tilde{\mu}(p) = \mu(p) \; \forall p \in M \cup W$ and $\tilde{\mu}(m) = \emptyset$, and $L(m)$ adds $\mu(m)$ to his waiting list. 
\item We remove a man $m$ by calling $\MFPAlg(W, M \setminus \{ m \}, \tilde{\mu}, \tilde{L})$.
We have $R(\tilde{\mu}) = R(\mu) \setminus \{ m, \mu(m) \}$, so $\tilde{\mu}|_{R(\tilde{\mu})}$ is also stable.
The waiting lists are compatible.
\item The other two cases are very similar.
\end{itemize}
Hence, the returned matching and waiting lists satisfy the same property in all cases.
Therefore, iteratively calling these routines preserves the stability of $\mu$.

Define the waiting lists $L$ such that they contain all acceptable women for men, i.e., $L(m) = P(m) \; \forall m \in M$, and are empty for women, $L(w) = \emptyset$. The matching $\mu = \{ (w, w) \mid w \in W\}$ is stable on the market restricted to $R(\mu) = W$. The waiting lists are compatible with $\mu$ on the large market $(M_f, W_f) = (M, W)$.\footnote{
Note that we can also define the waiting lists and preferences on a bigger market $(M_f, W_f) \supset (M, W)$. The choice of $(M_f, W_f)$ determines the men and women we can add with the above procedure.
}
Then, the matchings $\GSalgShort(M, W, P)$ and $\MFPAlg(M, W, \{ (w, w) \mid w \in W\}, L)$ agree. This can also be obtained by starting from the empty market $(\emptyset, \emptyset)$ and adding men and women one-by-one.
In general, however, mixing these routines may not result in the GS matching as illustrated with the following example. Consider the market over men $M = \{ m_1, m_2, m_3 \}$ and women $W = \{ w_1, w_2, w_3 \})$ with preferences $P(m_1) = [ w_2, w_3 ], P(m_2) = [ w_1, w_2 ], P(m_3) = [ w_2, w_1 ], P(w_1) = [m_3, m_2], P(w_2) = [m_2, m_1, m_3], P(w_3) = [m_1]$ (listing acceptable partners from best to worse). We first compute $\MFPAlg(M, W, \{ (w, w) \mid w \in W\}, L)$ and then remove man $m_1$. This does not coincide with the matching $\MFPAlg(M \setminus \{ m_1 \}, W, \{ (w, w) \mid w \in W\}, L) = \GSalgShort(M \setminus \{ m_1 \}, W, P)$.

As discussed in Section~\ref{sec:phdAlgorithm}, at intermediate iterations of the PhD algorithm, we are only interested in finding any stable matching to determine the set of matched persons. So we can use the above procedure for all but the last iteration. At the last iteration, we can rerun the GS algorithm from scratch in complexity $\mathcal{O}( (|A| + |C|) |S|)$.
The waiting lists can be implemented as max-heaps, where the cost to insert or extract the maximum is $O(\log n)$ for a heap of size $n$.
On the $(A, S)$ market, all students are added at the beginning and removed over iterations whenever they find an advisor, but no co-advisor.
If students initially propose to advisors (students are the ``men''), the total cost is $\mathcal{O}(|A| |S| \log |S|)$ since the waiting list of each advisor has size at most $|S|$ and each advisor proposes to each student at most once (after the first iteration when students are removed).
If advisors initially propose to students (advisors are the ``men''), the total cost is $\mathcal{O}(|A| |S|)$ (since we can find the most preferred student using a list rather than a heap in $\mathcal{O}(1)$). This is because advisors just keep proposing when students are removed.
If students initially propose to co-advisors (students are the ``men''), the total cost is $\mathcal{O}(|C| |S| \log |S|)$ since the waiting list of each co-advisor has size at most $|S|$ and each co-advisor proposes to each student at most once (after the first iteration when students are removed).
If co-advisors initially propose to students (co-advisors are the ``men''), the total cost is $\mathcal{O}(|C| |S| \log |C|)$ since the waiting list of each student has size at most $|C|$ and each student proposes to each co-advisor at most once (after the first iteration when students are removed).
Hence, the minimum achievable cost for all but the last iteration is either $\mathcal{O}((|A| + |C| \log |C|) |S|)$ or $\mathcal{O}((|A| + |C| \log |S|) |S|)$.

\section{Quotas}\label{app:quotas}
In this section, we show how to adapt the setting to quotas. 
We first show how to extend the two-sided setting to quotas without passing through the extended market construction \citep{roth1984evolution}. This has the advantage that we can enforce the same man and woman not to be matched more than once. In the extended market construction, each person is replicated according to his quota to obtain a new market without quotas. This discards the connection between replicas of the same person, so the same man and woman may be matched more than once.
Then, we do the same for three-sided markets based on two-sided markets. The main difference is the stability definition.
Finally, we proceed as in Section~\ref{sec:phdAlgorithm} to give an explicit proof for the setting with quotas.

\subsection{Two-sided markets}
For a review of many-to-many two-sided markets, see \citep{roth1991natural, echenique2004theory}. 
Roughly, person $p$ has $q_p \in \mathbb{N}$ spots, so can be matched up to $q_p$ times. A matching is unstable if there exist a man $m$ and woman $w$ such that they prefer to be matched together rather than to one of their current partners. This assumes that person $p$ has $q_p$ partners, where we fill empty spots with $p$ himself up to quota $q_p$.
In our setting, we assume that persons treat all their spots equally such that their preferences over groups of people are determined by their preferences over people.
In this setting, the algorithms presented in \citep{roth1991natural, echenique2004theory} simplify considerably.
In fact, we can find stable matchings by modifying the one-to-one GS algorithm. This is depicted in Algorithm~\ref{alg:galeShapleyQuotas}.
\begin{algorithm}[ht]
\DontPrintSemicolon
\SetAlgoLined
\KwIn{market $(M, W, P, Q)$, quotas $Q$}
\KwOut{matching $\mu$}
$\mu \gets \cup_{w\in W}\cup_{i=1}^{q_w} \{(w,w)\}$\; 
\While{there is a man $m$ with some unmatched spots, i.e., $|\mu(m)| < q_m$}{
 $w \gets \offerNext(P, m)$ \;
 \uIf{$w = m$}{
  $\mu \gets \mu \cup \{ (m, m) \}$ 
 }
 \ElseIf{$m >_w m' = \weakestMatch(\mu(w), P(w))$}{
    $\mu \gets \mu \setminus \{ (m', w) \} \cup \{ (m, w) \}$ 
 }
}
\Return{$\mu$}
\caption{Extended Gale--Shapley algorithm on many-to-many markets}
\label{alg:galeShapleyQuotas}
\end{algorithm}
Each woman $w$ is initially matched to herself $q_w$ times.
Whenever a man is not matched $q_m$ times (to women or himself), he proposes to his most preferred woman among the women to whom he did not propose yet (function $\offerNext$).
If that woman prefers the man to any partner (a man or herself) she is currently matched with, she unmatches from this partner and matches with the new man.
If a man has proposed to all his acceptable women, he fills the remaining spots with himself (since he is always acceptable to himself).
The function $\weakestMatch(\mu(w), P(w))$ returns the (currently) weakest partner of woman $w$ in $\mu(w)$, where, during the execution of the algorithm, we define $\mu(m) = \{ m \mid (m, w) \in \mu \}$ in terms of the temporary $\mu$ defined in the algorithm.

Stability of the matching and optimality are a direct consequence of the results from \cite{echenique2004theory}. Alternatively, the close resemblance to the one-to-one GS algorithm makes it possible to extend the existing one-to-one proofs for the GS algorithm \citep{gale1962college, knuth1997stable, sotomayor1990two}, which is more explicit. Similarly to Proposition~\ref{prop:twoSidedSameMatched}, it follows that the set of matched persons is also the same across all stable matchings.
We will use the optimality properties to prove stability of the PhD algorithm in the three-sided setting.
Another setting is when man $m$ and woman $w$ are allowed to match at most $q_{m, w}$ times. In this case, Algorithm~\ref{alg:galeShapleyQuotas} can be adapted such that a man proposes to his most preferred woman $w$ that did not reject him $q_{m, w}$ times already. Again, this results in stable and optimal matchings.

\subsection{Three-sided markets}
The setup with quotas in three-sided markets is very similar to the one in two-sided markets. Each person $p$ has $q_p \in \mathbb{N}$ spots. In particular, students have the same number of spots in both two-sided markets. A matching matches each person at most $q_p$ times.

For person $p$, define the preference between an individual $p_1$ and a set $K$ as $p_1 >_p K \iff \exists p_2 \in K: p_1 >_p p_2$. 
With this, we define stability for arbitrary quotas.
\begin{definition}
\label{def:blockingTripleQuotas}
The triple $(a, s, c) \in A \times S \times C$ blocks $\mu$ if there exists $k \in \{ 1, \dots, q_s \}$ such that:
\begin{align*}
\text{if $\mu(s)_k = (s, s)$ (unmatched):} &
&\begin{cases}
a >_s \mu(s)_{k, A} \\
s >_a \mu(a)_S
\end{cases}
\text{ and }
&\begin{cases}
c >_s \mu(s)_{k, C} \\
s >_c \mu(c)_S
\end{cases}, \\
\text{or if $\mu(s)_k \neq (s, s)$ (matched):} &
&\begin{cases}
a >_s \mu(s)_{k, A} \\
s >_a \mu(a)_S
\end{cases}
\text{ or }
&\begin{cases}
c >_s \mu(s)_{k, C} \\
s >_c \mu(c)_S
\end{cases}.
\end{align*}
Here, we assume that $\mu(s)$ is ordered (arbitrarily) so that $\mu(s)$ can be indexed.
Then $\mu(s)_{k, A}$ is the advisor that is matched to the $k$-th spot of $s$ in $\mu$.
\end{definition}

We have the following Fact. We omit the proof.
\begin{fact}
If $(a, s, c)$ blocks $\mu$ on the market $(A, S, C, P, Q)$, one of the following must hold:
\label{fact:blockingTripleImpliesPairQuotas}
\begin{align*}
\begin{cases}
a >_s \mu(s)_A \\
s >_a \mu(a)_S \\
\end{cases}
\textrm{or} \quad
\begin{cases}
a \in \mu(s)_A \\
c >_s \mu(s)_C \\
s >_c \mu(c)_S
\end{cases}.
\end{align*}
In other words, either $(a, s)$ blocks $\mu|_{(A, S, Q_{AS})}$ or $(s, c)$ blocks $\mu|_{(S, C, Q_{SC})}$, where the quotas $Q_{AS}, Q_{SC}$ on the markets $(A, S)$ and $(S, C)$ satisfy
\begin{itemize}
\item $Q_{AS}(p) = Q(p) \; \forall p \in A \cup \{ s \}$ and,
\item $Q_{SC}(p) = Q(p)\; \forall p \in C$, $Q_{SC}(s) \geq \countMatches(\mu, s)$, where $\countMatches(\mu, s)$ counts the number of times that $s$ is matched in $\mu$ excluding himself.
\end{itemize}
\end{fact}

In the PhD algorithm without quotas, a student is removed if he does not find a co-advisor. Here, when a student has a quota on the $(S, C)$ market and does not use all of his quota, we reduce his quota on the $(A, S)$ market by the number of unfilled spots.
More precisely, say a student matches with $k_a$ advisors. Then, he is assigned quota $k_a$ on the student-co-advisor market. If this student then matches to $k_c \leq k_a$ co-advisors, reduce his capacity on the advisor market by $k_a - k_c$ (instead of removing the student).
This is depicted in Algorithm~\ref{alg:phdAlgorithmAllWithQuotas}.
The function $\countMatches(\mu, s)$ counts the number of matches of $s$ excluding self-matches. We use dictionary notation to describe the quota function.
The function $Q.update$ updates quotas $Q$ by replacing the quotas of the specified persons only.
$Q|_{(A, S)}$ are the quotas restricted to the market $(A, S)$, $\tilde{Q}_{SC}$ are the quotas defined in the algorithm on the market $(S, C)$. 

\begin{algorithm}[ht]
\DontPrintSemicolon
\SetAlgoLined
\KwIn{PhD market $(A, S, C, P, Q)$}
\KwOut{stable matching $\mu$}
strictify the preferences $P_{AS}$ on $(A, S)$ and $P_{SC}$ on $(S, C)$ \;
\Do{$Q|_{(S, C)} \neq \tilde{Q}_{SC}$}{
$\mu_{AS} \gets \GSalgShort(M = S, W = A, P_{AS}, Q|_{(A, S)})$.\;
$\tilde{Q}_{SC} \gets \{ s: \countMatches(\mu_{AS}, s) \mid s \in S \} \cup \{ c: q_c \mid c \in C \}$ \;
$\mu_{SC} \gets \GSalgShort(M = S, W = C, P_{SC}, \tilde{Q}_{SC})$.\;
$Q.update(\{ s: Q(s) + \countMatches(\mu_{SC}, s) - \tilde{Q}_{SC}(s) \mid s \in S\})$
}
let matching $\mu$ such that student $s$ is matched to advisors $\mu_{AS}(s)$ and co-advisors $\mu_{SC}(s)$, fill the remaining spots of all persons with themselves\;
\Return{$\mu$}
\caption{PhD Market - PhD algorithm with quotas}
\label{alg:phdAlgorithmAllWithQuotas}
\end{algorithm}



Before proving stability, we establish a few properties. 
The proof is omitted for some of them.
At iteration $\tau$, define $q_s^{(\tau)}$ and $\tilde{q}_s^{(\tau)}$ as the capacities of student $s$ in the $(A, S)$ and $(S, C)$ market:
\begin{align*}
q_s^{(\tau)} &= q_s^{(\tau-1)} + \countMatches(\mu_{SC}^{(\tau-1)}, s) - \tilde{q}_s^{(\tau-1)} \quad \text{ for } \tau > 1 \;,\\
\tilde{q}_s^{(\tau)} &= \countMatches(\mu_{AS}^{(\tau)}, s) \;.
\end{align*}
Since $\countMatches(\mu, s)$ is upper-bounded by the capacity of $s$, we have $\tilde{q}_s^{(\tau)} \leq q_s^{(\tau)}$ and $q_s^{(\tau-1)} \geq q_s^{(\tau)} \geq \countMatches(\mu_{SC}^{(\tau-1)}, s)$ (for $\tau > 1$).
Note that $\tilde{q}_s^{(\tau)}$ can increase or decrease over iterations.

Based on this, Lemmas~\ref{lem:gsAddMen},~\ref{lem:studentStaysMatchedAtNextIteration} and~\ref{lem:coadvisorCannotDecrease} can be adapted as follows. 
Assume strict preferences.
As before, the match partners of $p$, i.e. $\mu(p), \mu_{AS}(p), \mu_{SC}(p)$, are ordered from best to worst.
For example, the matching $\mu_{AS}^{(\tau)}(s)$ of student $s$ at iteration $\tau$ is ordered in terms of $\geq_s$ from the most preferred to the least preferred advisor.
\begin{lemma}
\label{lem:gsAddMenQuotas}
Let a market $(M, W, P, Q)$ and quotas $Q_1, Q_2$ such that $Q_1(m) \leq Q(m) \; \forall m \in M, Q_1(w) = Q(w) \; \forall w \in W$ and $Q_2(m) = Q(m) \; \forall m \in M, Q_2(w) \leq Q(w) \; \forall w \in W$.
Let $\mu = \GSalgShort(M, W, Q)$ and $\mu_i = \GSalgShort(M, W, Q_i), i = 1, 2$.
Independently of whether men or women propose (but the same for $(\mu_1, \mu)$ and $(\mu_2, \mu)$ respectively), it holds that 
$\mu_1(m)_k \geq_m \mu(m)_k \; \forall m \in M, k = 1, \dots, Q_1(m), \mu_1(w)_k \leq_w \mu(w)_k \; \forall w \in W, k = 1, \dots, Q(w)$
and
$\mu_2(m)_k \leq_m \mu(m)_k \; \forall m \in M, k = 1, \dots, Q(m), \mu_2(w)_k \geq_w \mu(w)_k \; \forall w \in W, k = 1, \dots, Q_2(w)$.
\end{lemma}

\begin{lemma}
\label{lem:studentStaysMatchedAtNextIterationProjectAlg}
A student's matches (excluding self-matches) on the $(A, S, Q)$ market cannot decrease over iterations.
More precisely, for student $s$ with capacity $q_s^{(\tau)}$ at iteration $\tau \geq 2$, we have $\mu_{AS}^{(\tau)}(s)_{k} \geq_s \mu_{AS}^{(\tau-1)}(s)_{k} \; \forall k = 1, \dots, q_s^{(\tau)}$.
This implies that 
\[
\countMatches(\mu_{AS}^{(\tau)}, s) \geq \min(q_s^{(\tau)}, \countMatches(\mu_{AS}^{(\tau-1)}, s)) \;.
\]
\end{lemma}
\begin{proof}
Since $q_s^{(\tau)} \leq q_s^{(\tau-1)} \; \forall s$, it follows from Lemma~\ref{lem:gsAddMenQuotas}. Intuitively, the $q_s^{(\tau)}$ spots of student $s$ face less competition from other students than in the previous iteration.
\end{proof}

\begin{lemma}
\label{lem:coadvisorCannotDecreaseProjectAlg}
A co-advisor's matches on the $(S, C, Q)$ market cannot decrease over iterations.
More precisely, for co-advisor $c$ at iteration $\tau > 1$, $\mu_{SC}^{(\tau)}(c)_{k} \geq_c \mu_{SC}^{(\tau-1)}(c)_{k} \; \forall k = 1, \dots, q_c$. 
\end{lemma}
\begin{proof}
We first prove that $\countMatches(\mu_{SC}^{(\tau-1)}, s) \leq \tilde{q}_s^{(\tau)}$. By Lemma~\ref{lem:studentStaysMatchedAtNextIterationProjectAlg}, 
\begin{align*}
    \tilde{q}_s^{(\tau)} = \countMatches(\mu_{AS}^{(\tau)}, s) 
    &\geq \min(q_s^{(\tau)}, \countMatches(\mu_{AS}^{(\tau-1)}, s)) \\
    &\geq \countMatches(\mu_{SC}^{(\tau-1)}, s)\;.    
\end{align*}
Now apply Lemma~\ref{lem:gsAddMenQuotas} by first matching students with capacity $\countMatches(\mu_{SC}^{(\tau-1)}, s) \leq \tilde{q}_s^{(\tau)}$ on the $(S, C)$ market.
This returns a matching that is the same as at iteration $\tau-1$ (except for some self-matches).
Then increase their capacities to $\tilde{q}_s^{(\tau)}$. 
Co-advisors cannot worsen, students cannot get better matches.
\end{proof}
Note that for student $s$, it also holds that $\mu_{SC}^{(\tau)}(s)_{k} \leq_s \mu_{SC}^{(\tau-1)}(s)_{k} \; \forall k = 1, \dots, \countMatches(\mu_{SC}^{(\tau-1)}, s)$. 
In other words, a student's matches (excluding self-matches) on the $(S, C, Q)$ market cannot increase over iterations.
We do not need this fact to prove stability though.
%
Proposition~\ref{prop:stableUnderStrictification} continues to hold.
Finally, we can prove stability.

\begin{theorem}\label{thm:stablePhDAlgorithmQuotas}
The PhD algorithm with quotas returns a stable matching.
\end{theorem}
\begin{proof}
By Proposition~\ref{prop:stableUnderStrictification}, it is sufficient to show that the matching $\mu$ is stable under the strictified preferences (first line of the algorithm).
We proceed by contradiction.
The algorithm cannot return an individually irrational matching since the two-sided GS algorithms do not.

Assume that $(a, s, c) \in A \times S \times C$ blocks $\mu$. Let $T$ the last iteration of the PhD algorithm and $Q^{(\tau)}$ the quotas at iteration $\tau$ on the market $(A, S, C)$.
If $s$'s capacity has not changed, $q_s^{(1)} = q_s^{(T)}$, Fact~\ref{fact:blockingTripleImpliesPairQuotas} implies that either $(a, s)$ blocks $\mu|_{(A, S, Q^{(T)}|_{(A, S)})}$ or $(s, c)$ blocks $\mu|_{(S, C, Q^{(T)}|_{(S, C)})}$. 
Since $Q^{(T)}|_{(S, C)} = \tilde{Q}_{SC}^{(T)}$ at the final iteration, this contradicts the stability of the GS algorithm on one of the two markets.

If $s$'s capacity has changed, $q_s^{(T)} < q_s^{(1)}$, $s$ matched to $\tilde{q}_s^{(\tau)} = \countMatches(\mu_{AS}^{(\tau)}, s)$ advisors and less co-advisors at some iteration $\tau$ (since $q_s^{(\tau)}$ can only decrease).
Consider $c$'s weakest student $\tilde{s} = \mu_{SC}^{(\tau)}(c)_{q_c}$ at that time. Since $c$ and $s$ are mutually acceptable, $c$ matches with $\tilde{s} >_c s$ at that iteration, independently of whether students or co-advisors propose.
By Lemma~\ref{lem:coadvisorCannotDecreaseProjectAlg}, $c$'s weakest student satisfies $\mu(c)_{S, q_c} = \mu_{SC}^{(T)}(c)_{q_c} \geq_c \mu_{SC}^{(\tau)}(c)_{q_c} = \tilde{s} >_c s$. Hence $\tau >_c s \; \forall \tau \in \mu(c)_S$ and $(a, s, c)$ cannot be a blocking triple.
\end{proof}

\section{Challenges with incompatible advisor pairs}\label{sec:incompatibleAdvisorPairs}
We show by counterexamples that incompatible pairs cannot be incorporated through simple modifications of the PhD algorithm.
Consider the case when only some advisor pairs $K \subset A \times C$ are compatible. 
A matching $\mu$ has to be a subset of $\left\{ (a, s, c) \in A \times S \times C \mid (a, c) \in K \right\}$ rather than $A \times S \times C$. Let $K(a) = \{ c \mid (a, c) \in K \}$ the co-advisors that are acceptable to $a$ and $K(c) = \{ a \mid (a, c) \in K \}$ the advisors that are acceptable to $c$.
Suppose that we adapt the PhD algorithm such that a student who is matched to an advisor, at a given iteration, only considers co-advisors who are compatible with his matched advisor. This is clearly necessary since the algorithm may otherwise match advisors to incompatible co-advisors.
The left part of Figure~\ref{fig:counterExamplePhDAlgIncompatibilities} shows that this modification generally does not work.
\begin{figure}
    \centering
    \includestandalone[width=\textwidth,height=.12\textheight,keepaspectratio]{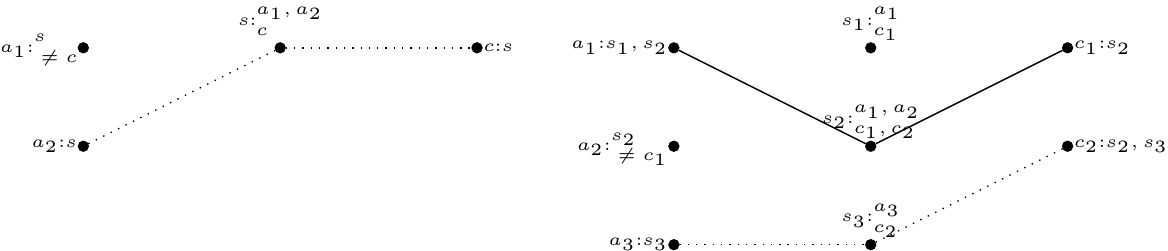}
    \caption{
    The PhD Algorithm fails when some advisor pairs are incompatible.
    Note that this reasoning works independently of which side proposes on the two-sided markets.
    \textbf{Left:}
    When the PhD Algorithm (Algorithm~\ref{alg:phdAlgorithm}) does not take incompatible advisor pairs into account, it returns the matching $\{ (a_1, s, c) \}$, which is not valid since $a_1$ and $c$ are incompatible.
    If a student only considers co-advisors who are compatible with his matched advisor, $s$ is removed at the first iteration and the algorithm returns the unstable matching $\{ \}$ which is blocked by $(a_2, s, c)$.
    If the student removes advisor $a_1$ from his list of acceptable advisors after the first iteration, the algorithm produces the stable matching $\{ (a_2, s, c) \}$ (in dotted style).
    \textbf{Right:} However, this modification does not work for this example where the PhD algorithm returns the unstable matching $\{ (a_1, s_2, c_1) \}$ which is blocked by $(a_3, s_3, c_2)$.
    At the first iteration, the temporary matching is $\{ (a_2, s_2, c_2) \}$, $s_1$ removes $a_1$ (from his list) and $s_3$ removes advisor $a_3$.
    The algorithm terminates after the second iteration.
    The reason this happens is that $a_2$ is not compatible with $c_1$, so $s_2$ matches with $c_2$, so that $s_3$ removes $a_3$ from his list of acceptable advisors.
    At the second iteration, $s_2$ matches with advisor $a_1$ who is compatible with $s_2$'s more preferred co-advisor $c_1$.
    If $s_3$ were still part of the market, it could match to $(a_3, c_2)$.
    }
    \label{fig:counterExamplePhDAlgIncompatibilities}
\end{figure}
Whenever a student cannot find a co-advisor among those that are compatible with his matched advisor $a$, the student is removed. However, it could be that $a$ is not compatible with many co-advisors. If $s$ changed his advisor, he could find a match.
So the idea is to (greedily) remove $a$ from $s$'s preferences in all future iterations (i.e. student $s$ marks $a$ as unacceptable), so $s$ cannot match again with $a$.
When the student is matched to advisor $a$ and cannot find a co-advisor at an iteration, he marks advisor $a$ as unacceptable starting from that iteration.
Then, this modified algorithm applied to the left part of Figure~\ref{fig:counterExamplePhDAlgIncompatibilities} works.
However, the right part of this figure shows that it generally does not work.
To understand why this happens, note that the PhD algorithm first addresses the $(A, S)$ market and then the $(S, C)$ market. It treats them as separate entities with very limited interaction, only through the set $S_m$ of matched students and the set of removed students $S_u$ at each iteration.
A different approach could be that when a student finds an advisor, he immediately tries to find a co-advisor as well, and this is repeated until all students are matched.
\cite{zhong2019Cooper} provide an algorithm based on this idea (additionally assuming complete preferences).
However, their algorithm appears to be incorrect as we discuss now by providing a counterexample.

\subsection{Remark on [Zhong and Bai, 2019]}
\label{app:incompatibleAdvisorPairsOtherAlg}
The authors consider the setting where advisors can be incompatible, and preferences are complete and strict. Their algorithm supposedly computes a stable matching when advisor pairs can be incompatible.
We first write down the algorithm they describe by letting persons propose one after the other rather than all at once (the two ways are equivalent).
Then, we give a counterexample to show that it does not work with incompatible advisor pairs. In the case when all advisor pairs are compatible, it reduces to a special case of the PhD algorithm.

As discussed above, their algorithm can be motivated from the PhD algorithm. Rather than running two separate two-sided markets, the idea is to interleave them.
Pick any currently unmatched advisor and let him propose to the most preferred student that did not reject him. 
If the student is already matched to an advisor that is worse than the proposing advisor and the proposing advisor is compatible with the student's current co-advisor, the student rejects the old advisor and accepts the new advisor, keeping the co-advisor. Otherwise, the student rejects the proposing advisor.
If the student is not already matched, he proposes to and matches with the most preferred compatible co-advisor who prefers the student to his current student. If no such co-advisor exists, the student rejects the advisor. This algorithm is depicted in Algorithm~\ref{alg:phdAlgorithmIncompatiblePairs}.
\begin{algorithm}[ht]
\DontPrintSemicolon
\SetAlgoLined
\KwIn{PhD market $(A, S, C, P)$, quotas $Q = 1$}
\KwOut{matching $\mu$}
$\mu \gets \{ (s, s, s) \mid s \in S \} \cup \{ (c, c, c) \mid c \in C \}$  \;
\While{there is an unmatched advisor $a$}{
 $s \gets \offerNext(P, a)$ \;
 \uIf{$s = a$}{
   $\mu \gets \mu \cup \{ (a, a, a) \}$
 }
 \Else{
 \uIf{$s$ is matched to $(\tilde{a}, \tilde{c}) \in A \times C$}{
    \label{lst:phdAlgIncompStudMatched}
    \uIf{$a >_s \tilde{a}$ and $a \in K(\tilde{c})$}{
        $\mu \gets \mu \setminus \{ (\tilde{a}, s, \tilde{c}) \} \cup \{ (a, s, \tilde{c}) \}$ \;
        Reject $\tilde{a}$
    }
    \Else{
    Reject $a$ \label{lst:phdAlgIncompFailureMode}
    }
 }
 \Else{
    \label{lst:phdAlgIncompStudUnmatched}
    $\tilde{C} \gets \{ c : s >_c \mu(c)_S \land c \in K(a) \}$ \; \label{lst:phdAlgIncompCTilde}
    \uIf{$\tilde{C} = \emptyset$}{
        Reject $a$
    }
    \Else{
        $c \gets \bestAmong(\tilde{C}, P(s))$ \;
        $\mu \gets \mu \setminus \{ (\mu(c)_A, \mu(c)_S, c) \} \cup \{ (a, s, c) \}$
    }
 }
 }
}
\Return{$\mu$}
\caption{Algorithm by \cite{zhong2019Cooper}}
\label{alg:phdAlgorithmIncompatiblePairs}
\end{algorithm}
The function $\offerNext(P, p)$ returns the most preferred person of $p$ who did not reject $p$ yet. 
When $p_1$ rejects $p$, this can be implemented by removing $p$ from $P(p_1)$, the preferences of $p_1$.
When $s$ is unmatched, $s$ only considers co-advisors in $\tilde{C}$ which ensures that $(a, s, c)$ can be matched because $c$ prefers $s$ to his current match. Student $s$ could instead propose to the most preferred co-advisor who did not reject him yet and the co-advisor only accepts if he prefers $s$ to his current match.
The function $\bestAmong(\tilde{C}, P(s))$ picks $s$'s most preferred co-advisor among $\tilde{C}$.
Note that when $s$ matches with co-advisor $c$ and this unmatches $(\tilde{a}, \tilde{s}, c)$, the advisor $\tilde{a}$ is not (immediately) rejected by $s$. Advisor $\tilde{a}$ can propose to $\tilde{s}$ again and $\tilde{s}$ can find a new co-advisor. The algorithm returns a valid matching in the sense that each person is matched to himself or exactly one acceptable couple of persons of the other two sides.
Its running time is $\mathcal{O}((|A| + |C|) |S|)$.

The algorithm fails when some advisor pairs are incompatible as shown by the counterexample in Figure~\ref{fig:counterExampleOtherAlgIncompatibilities}, even under complete preferences and when $|A| = |S| = |C|$.
\begin{figure}
    \centering
    \includestandalone[width=\textwidth,height=.15\textheight,keepaspectratio]{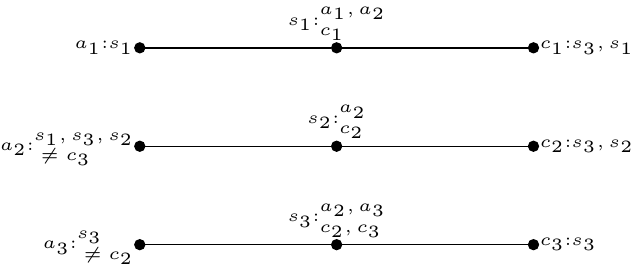}
    \caption{
    Algorithm~\ref{alg:phdAlgorithmIncompatiblePairs} from \citep{zhong2019Cooper} fails when some advisor pairs are incompatible. It returns the shown matching $\{ (a_1, s_1, c_1), (a_2, s_2, c_2), (a_3, s_3, c_3) \}$ which is blocked by $(a_2, s_3, c_2)$. The compatible advisor pairs are $K = A \times C \setminus \{ (a_3, c_2), (a_2, c_2) \}$.
    Even if people propose in parallel (as originally presented in \citep{zhong2019Cooper}), it does not work.
    Call a round when all advisors propose to students and the students propose to co-advisors (if they are unmatched).
    After the first round, the matching is $\{ (a_1, s_1, c_1), (a_3, s_3, c_3) \}$ because $(a_3, c_2) \notin K$. In the second round, $s_3$ rejects $a_2$ since $(a_2, c_3) \notin K$. In the third round, the matching becomes $\{ (a_1, s_1, c_1), (a_2, s_2, c_2), (a_3, s_3, c_3) \}$.
    Since everyone is matched, the matching stays the same if the unacceptable persons are appended at the end of each person's preferences as acceptable persons. Therefore, this is also a counterexample when preferences must be complete.
    The matching $\{ (a_1, s_1, c_1), (a_2, s_3, c_2) \}$ is stable.
    As an aside, the PhD algorithm (Algorithm~\ref{alg:phdAlgorithm}) happens to produce the stable matching $\{ (a_1, s_1, c_1), (a_2, s_3, c_2) \}$ on this example.
    }
    \label{fig:counterExampleOtherAlgIncompatibilities}
\end{figure}
On the other hand, when all advisor pairs are compatible, one can see that this algorithm reduces to the special case of the PhD algorithm with advisors proposing to students and students proposing to co-advisors.
Though, our PhD algorithm allows to find a richer family of stable matchings (see Theorem~\ref{thm:invariancePhDAlg}).
Given these challenges, we wonder if the setting with incompatible advisors is NP-complete.

%
\section{Multi-sided markets}\label{sec:nsidedMarkets}
We have a necessary condition for blocking tuples in $n$-sided markets.
\begin{fact}\label{fact:blockingTupleImpliesBlockingPair_app}
If $(s_1, \dots, s_n) \in S_1 \times \dots \times S_n$ blocks $\mu$, there exists some $k = 1, \dots, n-1$ such that
\begin{align*}
\begin{cases}
\mu_{j-1}(s_{j-1}) = s_j \; \forall 2 \leq j \leq k \\
s_{k+1} >_{s_k} \mu_{k}(s_k) \\
s_k >_{s_{k+1}} \mu_k(s_{k+1}) \\
\end{cases}.
\end{align*}
Define $M_1 = S_1$ and, for $k > 1$, $M_k = \{ s_k \in S_k \mid \mu_k(s_k) \neq s_k \}$ is the set of matched persons among $S_k$.
Then, at least one of the matchings $\mu_k|_{(M_k, S_{k+1})}$ must be blocked by $(s_k, s_{k+1})$.
\end{fact}
\begin{proof}
Say $H_k$ is the statement that the above does not hold for a given $k$.
If there exists no $k$ for which the above holds, $H_1, \dots, H_{n-1}$ all hold. Statement $H_1$ means that $(s_1, s_2) \in \mu_1$. $H_2$ (together with $H_1$) then implies that $(s_2, s_3) \in \mu_2$. By iterating, we find that $(s_k, s_{k+1}) \in \mu_k \; \forall k = 1, \dots, n-1$, hence $(s_1, \dots, s_n) \in \mu$, so $(s_1, \dots, s_n)$ cannot be a blocking tuple.
Hence $\mu_k$ must be blocked by $(s_k, s_{k+1})$ for some $k$. 
If $k > 1$, $\mu_{k-1}(s_{k-1}) = s_k$ implies that $s_k$ is matched in $\mu$, hence $s_k \in M_k$. (The case $k=1$ is immediate.)
Clearly, $(s_k, s_{k+1})$ also blocks $\mu_k|_{(M_k, S_{k+1})}$.
\end{proof}

Again, stable matchings exist and the extension of the PhD algorithm to $n$ sides is depicted in Algorithm~\ref{alg:phdAlgorithmNSided}.
\begin{algorithm}[ht]
\DontPrintSemicolon
\SetAlgoLined
\KwIn{PhD market $\brb{(S_1, \dots, S_n), (P_1, \dots, P_{n-1})}$}
\KwOut{stable matching $\mu$}
Strictify the preferences $P_1, \dots, P_{n-1}$\; 
Set $\tau \gets 0$ and $S_k^{(1)} \gets S_k$ for all $k=1,\ldots,n$ \;
\Do{$\bigcup_{k=2}^{n-1} U_k^{(\tau)} \neq \emptyset$}{
$\tau \gets \tau + 1$ \;
$M_1^{(\tau)} \gets S_1^{(\tau)}$ \;
\For{$k = 1, \dots, n-1$}{
    $\mu_k^{(\tau)} \gets \GSalgShort( M_k^{(\tau)}, S_{k+1}^{(\tau)}, P_k )$ \tcp{matching on a submarket of $(S_k, S_{k+1})$} 
    let $M_{k+1}^{(\tau)}$ the matched persons among $S_{k+1}^{(\tau)}$ in $\mu_k^{(\tau)}$ \;
    let $U_{k}^{(\tau)}$ the unmatched persons among $S_{k}^{(\tau)}$ in $\mu_k^{(\tau)}$ \;
}
\For{$k = 2, \dots, n-1$}{
    $S_k^{(\tau+1)} \gets S_k^{(\tau)} \setminus U_k^{(\tau)}$
}
}
match the persons according to matchings $\mu_1^{(\tau)}, \dots, \mu_{n-1}^{(\tau)}$, the removed and unmatched persons stay single\;
\Return{$\mu$}
\caption{$n$-sided PhD market}
\label{alg:phdAlgorithmNSided}
\end{algorithm}
The superscript $(\tau)$ denotes the iteration number and the subscript $k$ refers to a submarket of the $k$-th market $(S_k, S_{k+1})$.
The algorithm progresses by consecutively finding a matching on a submarket of $(S_k, S_{k+1})$ for $k = 1, \dots, n-1$. In the $k$-th market, only the persons in $S_{k}$ participate who found a match in the $(k-1)$-th market. If such a matched person cannot find a match on the $(k+1)$-th two-sided market, this person will also not find a match in later iterations and he is therefore removed unless he is part of the first or last side. The algorithm iterates until no person is removed. All removed and unmatched persons stay single.
The algorithm produces a valid matching, i.e., no person is partially matched or matched more than once. 
Indeed, a person can be matched at most once since the GS algorithm matches the same person at most once on each two-sided market, in particular at the final iteration.
Suppose that a person is partially matched, then he must be returned as part of some tuple $(s_1, \dots, s_k) \subset S_1 \times \dots \times S_k$ with $1 < k < n$. This can only happen if $s_k$ could not find a match on the market $\brb{S_k^{(T)}, S_{k+1}^{(T)}}$, hence $s_k \in U_k^{(T)}$ at the final iteration $T$ of the PhD algorithm. Then, $\bigcup_{k=2}^{n-1} U_k^{(T)}$ could not have been empty. 
For the stability proof, we assume that the proposing side is fixed over iterations in each two-sided market. As for three-sided markets, this does not matter. In fact, we can replace the GS algorithms by any other (individually rational) stable matching algorithms since the set of matched persons is the same at each iteration. The matching algorithms used at the final iteration affect the properties of the matching, e.g., optimality.
When the GS algorithms are used on the two-sided markets, the PhD algorithm can again be efficiently implemented in $\mathcal{O}(\sum_{k=1}^{n-1} |S_k| |S_{k+1}|)$.

Before proving stability, we establish a property that extends Lemmas~\ref{lem:studentStaysMatchedAtNextIteration} and~\ref{lem:coadvisorCannotDecrease}.
First of all, note that $M_1^{(\tau)}$ and $S_n^{(\tau)}$ never change over iterations and $M_k^{(\tau)} \subset S_k^{(\tau)} \; \forall \tau$ and $ k = 1, \dots, n$. 
Furthermore, the sets of available persons cannot increase over iterations: $S_k^{(\tau+1)} \subset S_k^{(\tau)} \forall \tau $ and $ k = 2, \dots, n-1$.
For $s_k \in S_k$, we call $(S_{k-1}, S_k)$ $s_k$'s left market since it is the first market $s_k$ participates in, and $(S_k, S_{k+1})$ its right market.
Similarly to the three-sided case, a person's match cannot decrease over iterations on its left market, unless the person cannot find a match on its right market and is removed.
\begin{lemma}
\label{lem:matchedPersonsMatchCannotDecrease}
The match of person $s_{k+1} \in S_{k+1}$ on the market $(S_k, S_{k+1})$ cannot decrease.
More precisely, let $\mu_k^{(\tau)}$ the matching at iteration $\tau$ as in the algorithm. Then, $\mu_k^{(\tau)}(s_{k+1}) \geq_{s_{k+1}} \mu_k^{(\tau-1)}(s_{k+1}) \; \forall \tau \geq 2,$ with $s_{k+1} \in S_{k+1}^{(\tau)},$ and $ k = 1, \dots, n-1$.
Moreover, $M_k^{(\tau-1)} \setminus U_k^{(\tau-1)} \subset M_k^{(\tau)} \; \forall k = 2, \dots, n-1$.
\end{lemma}
\begin{proof}
The second statement follows from the first. Indeed, suppose $s_k \in M_k^{(\tau-1)} \setminus U_k^{(\tau-1)}$ for some $k = 2, \dots, n-1$ and the first statement holds for $k-1$. Since $M_k^{(\tau-1)} \subset S_k^{(\tau-1)}$, $s_k \in S_k^{(\tau-1)} \setminus U_k^{(\tau-1)} = S_k^{(\tau)}$.
$s_k \in M_k^{(\tau-1)}$ means that $s_k$ is matched in $\mu_{k-1}^{(\tau-1)}$, i.e., $\mu_{k-1}^{(\tau-1)}(s_k) >_{s_k} s_k$. Since $s_k \in S_k^{(\tau)}$, we have by the first statement, $\mu_{k-1}^{(\tau)}(s_k) \geq_{s_k} \mu_{k-1}^{(\tau-1)}(s_k) >_{s_k} s_k$, hence $s_k$ is also matched in $\mu_{k-1}^{(\tau)}$, i.e., $s_k \in M_k^{(\tau)}$.

For fixed $n$, we prove the first argument by induction over $k$. 
For the base case $k=1$, $\mu_1^{(\tau)} = \GSalgShort(M_1^{(\tau)}, S_2^{(\tau)})$. Since $M_1^{(\tau)}$ does not change over iterations and $S_2^{(\tau)}$ does not increase, Lemma~\ref{lem:gsAddMen} implies that $\mu_1^{(\tau)}(s_2) \geq_{s_2} \mu_1^{(\tau-1)}(s_2) \; \forall s_{2} \in S_{2}^{(\tau)}, \tau \geq 2$.
Assume that the statement holds for $k-1$, then we prove that it holds for $k < n$.
Consider the $\tau$-th iteration. First run the $k$-th GS algorithm on the market $(M_k^{(\tau-1)} \cap M_k^{(\tau)}, S_{k+1}^{(\tau-1)})$ to obtain the matching $\lambda_k^{(\tau)}$. By induction, $M_k^{(\tau-1)} \cap M_k^{(\tau)} = M_k^{(\tau-1)} \setminus U_k^{(\tau-1)}$ (which depends on the hypothesis for $k-1$). Since the persons $U_k^{(\tau-1)}$ are unmatched in $\mu_k^{(\tau-1)}$, $\lambda_k^{(\tau)}$ is the same as the matching $\mu_k^{(\tau-1)}$ restricted to the market $(M_k^{(\tau-1)} \setminus U_k^{(\tau-1)}, S_{k+1}^{(\tau-1)})$.
Now add the remaining $M_k^{(\tau)} \setminus M_k^{(\tau-1)}$ and remove $S_{k+1}^{(\tau-1)} \setminus S_{k+1}^{(\tau)}$ to obtain the market $(M_k^{(\tau)}, S_{k+1}^{(\tau)})$. By applying Lemma~\ref{lem:gsAddMen} twice, we see that $\mu_k^{(\tau)}(s_{k+1}) \geq_{s_{k+1}} \lambda_k^{(\tau)}(s_{k+1}) = \mu_k^{(\tau-1)}(s_{k+1}) \; \forall \tau \geq 2, s_{k+1} \in S_{k+1}^{(\tau)}$.
\end{proof}

Finally, we prove that the PhD algorithm in Algorithm~\ref{alg:phdAlgorithmNSided} produces a stable matching.
\begin{customthm}{\ref{thm:stablePhDAlgorithmNSided}}
The $n$-sided PhD Algorithm returns a stable matching.
\end{customthm}
\begin{proof}
By an analogue of Proposition~\ref{prop:stableUnderStrictification}, we can assume strict preferences.
The matching must be individually rational since the GS algorithms never match unacceptable partners on the two-sided markets.
Suppose that $(s_1, \dots, s_n) \in S_1 \times \dots \times S_n$ is a blocking tuple.
Fact~\ref{fact:blockingTupleImpliesBlockingPair_app} means that $\mu_k|_{(M_k^{(T)}, S_{k+1})}$ is blocked by $(s_k, s_{k+1})$ for some $k$.
Since $\mu_{j-1}(s_{j-1}) = s_j \; \forall 2 \leq j \leq k$, $s_2, \dots, s_k$ must all be matched and cannot have been removed.
Assume that $s_{k+1}$ was not removed during the PhD algorithm. 
Since $s_{k+1} \in S_{k+1}^{(T)}$, 
we have that $(s_k, s_{k+1})$ also blocks $\mu_k|_{(M_k^{(T)}, S_{k+1}^{(T)})}$.
This contradicts the stability of the GS algorithm in the $k$-th market $(M_k^{(T)}, S_{k+1}^{(T)})$ at the final iteration.
If $s_{k+1}$ was removed during the PhD algorithm, $k+1 < n$ and it must be that $s_{k+1}$ participated in the market $(M_{k+1}^{(\tau)}, S_{k+2}^{(\tau)})$ at some iteration $\tau$ when $s_{k+1}$ was removed.
Independently of which side proposes, this means that $s_{k+2}$ was already removed before or matched to a better $\tilde{s}_{k+1} >_{s_{k+2}} s_{k+1}$.
Suppose that $s_{k+2}$ is not removed in the PhD algorithm.
By Lemma~\ref{lem:matchedPersonsMatchCannotDecrease}, $s_{k+2}$ is matched in $\mu_{k+1}$ to $\mu_{k+1}(s_{k+2}) \geq_{s_{k+2}} \tilde{s}_{k+1} >_{s_{k+2}} s_{k+1}$.
Hence $(s_{k+1}, s_{k+2}) \notin \mu$ and it also does not block $\mu_{k+1}$ (since this would require $\mu_{k+1}(s_{k+2}) \leq_{s_{k+2}} s_{k+1}$). So $(s_1, \dots, s_n)$ cannot block $\mu$.
If $s_{k+2}$ was removed, apply the reasoning to $s_{k+2}$ instead of $s_{k+1}$.
It must eventually stop because no student among $S_n$ is removed in the algorithm.
\end{proof}

\section{Numerical simulation details}\label{app:numericalSimulation}
We generate the dataset with the following generative model. There is a total of number of research fields. Each person uniformly chooses a number of research fields they are interested in and samples them uniformly without replacement from the available research fields.
Let $R(p)$ the multi-one-hot encoding of the research fields of person $p$. For example, if the available research fields are \verb|T1, T2, T3, T4| and person $p$ is interested in \verb|T2, T4|, then $R(p) = (0, 1, 0, 1)$.
Person $p$ computes the research overlap with person $\tilde{p}$ as $o_{p \tilde{p}} = R(p)^T R(\tilde{p})$ and assigns person $\tilde{p}$ a ranking $r_{p \tilde{p}} = o_{p \tilde{p}} + \sigma \cdot U$ for $U \sim \text{Unif}[0, 1]$. We add random jitter parametrized by $\sigma$ to make the preferences less structured. Indeed, without jitter, if person $p$ ranks person $\tilde{p}$ highly, person $\tilde{p}$ is also likely to rank person $p$ highly (under a suitable probabilistic model on the preferences as we have it here).
We describe the parameters and their values in Table~\ref{tab:simulationParams}. With these parameters, $r_{p \tilde{p}} \in [-10, 10]$, so the jitter $\sigma = 3.4$ has an influence.

\begin{table}[h!]
\begin{center}
\begin{tabular}{ l l l }
Parameter name & Description & Value \\ 
\hline
\hline
\verb|nb_advisors| & Number of advisors & 350 \\
\verb|nb_students| & Number of students & 620 \\
\verb|nb_coadvisors| & Number of co-advisors & 500 \\
\hline
\hline
\verb|total_nb_fields| & Total number of research fields & 30 \\  
\verb|min_choosable_fields| & Minimum number of research fields per person & 5 \\
\verb|max_choosable_fields| & Maximum number of research fields per person & 10 \\
\verb|random_jitter| ($\sigma$) & Random jitter added to preferences & 3.4 \\
\hline
\hline
\verb|advisor_min_nb_prefs| & Minimum number of students an advisor can rank & 10 \\
\verb|student_min_nb_prefs_adv| & Minimum number of advisors a student can rank & 5 \\
\verb|student_min_nb_prefs_coadv| & Minimum number of co-advisors a student can rank & 5 \\
\verb|coadvisor_min_nb_prefs| & Minimum number of students a co-advisor can rank & 5 \\
\hline
\verb|advisor_max_nb_prefs| & Maximum number of students an advisor can rank & 30 \\
\verb|student_max_nb_prefs_adv| & Maximum number of advisors a student can rank & 10 \\
\verb|student_max_nb_prefs_coadv| & Maximum number of co-advisors a student can rank & 10 \\
\verb|coadvisor_max_nb_prefs| & Maximum number of students a co-advisor can rank & 30 \\
\end{tabular}
\end{center}
\caption{Parameters used to generate the preferences.}
\label{tab:simulationParams}
\end{table}

We run the PhD algorithm with students proposing to advisors and co-advisors in both two-sided markets.
If the proposing side changes, the final matching changes, but the number of matches and blocking triples is the same, see the invariance property discussed in Theorem~\ref{thm:invariancePhDAlg}. Therefore, the shown figures do not change.
Since the PhD algorithm is deterministic given strict preferences, we average over 40 datasets (generated using 40 seeds). The number of iterations may differ between different datasets, so at each iteration, we average over those runs that reached that iteration.
We can have at most $\min(350, 620, 500) = 350$ number of matches and we see that we find a good number of matches (roughly 230).
If we change these parameters (to reasonable values), we observe the same qualitative behavior. In general, the number of iterations is lower with less persons.
The code is available in the supplementary material. The experiments can run on a normal computer on the CPU. 

\end{document}